\documentclass[
    reprint,
    prx,
    amsmath,
    amssymb,
    aps,
    superscriptaddress,
    longbibliography,
    floatfix
]{revtex4-2}

\usepackage{graphicx}
\usepackage{dcolumn}
\usepackage{bm}
\usepackage{placeins}

\usepackage{etoolbox}
\usepackage{comment}
\usepackage{multirow}
\usepackage{textcomp}
\usepackage{ragged2e}

\usepackage{quantikz}
\usepackage{amsthm}
\usepackage{braket}

\usepackage{caption}
\usepackage{subcaption}
\usepackage{xcolor}
\usepackage[colorlinks=true,linkcolor=blue,citecolor=blue,urlcolor=blue]{hyperref}

\newcommand{\abs}[1]{\bigl|#1\bigr|}
\newcommand{\norm}[1]{\left\lVert#1\right\rVert}

\newtheorem{theorem}{Theorem}
\newtheorem{proposition}{Proposition}

\newtheorem{definition}{Definition}

\newtheorem{remark}{Remark}
\newtheorem*{restatement}{Restatement}

\DeclareMathOperator{\Tr}{Tr}

\begin{document}

\title{A rigorous hybridization of variational quantum eigensolver with classical neural network}

\author{Minwoo James Kim}
 \affiliation{Dept. of Computer Science and Engineering, Seoul National University, Seoul 08826, South Korea}
 \affiliation{NextQuantum, Seoul National University, Seoul 08826, South Korea}
\affiliation{Automation and Systems Research Institute, Seoul National University, Seoul 08826, South Korea}

\author{Kyoung Keun Park}%
\affiliation{Dept. of Computer Science and Engineering, Seoul National University, Seoul 08826, South Korea}
\affiliation{NextQuantum, Seoul National University, Seoul 08826, South Korea}
\affiliation{Automation and Systems Research Institute, Seoul National University, Seoul 08826, South Korea}

\author{Kyungmin Lee}
 \affiliation{Dept. of Computer Science and Engineering, Seoul National University, Seoul 08826, South Korea}
 \affiliation{NextQuantum, Seoul National University, Seoul 08826, South Korea}
\affiliation{Automation and Systems Research Institute, Seoul National University, Seoul 08826, South Korea}

\author{Jeongho Bang\textsuperscript{*}}
\affiliation{Institute for Convergence Research and Education in Advanced Technology, Yonsei University, Seoul 03722, Republic of Korea}
\affiliation{Department of Quantum Information, Yonsei University, Incheon 21983, South Korea}
\thanks{Corresponding authors: \href{jbang@yonsei.ac.kr}{jbang@yonsei.ac.kr} and \href{taehyun@snu.ac.kr}{taehyun@snu.ac.kr}}

\author{Taehyun Kim\textsuperscript{*,}}
\affiliation{Dept. of Computer Science and Engineering, Seoul National University, Seoul 08826, South Korea}
\affiliation{NextQuantum, Seoul National University, Seoul 08826, South Korea}
\affiliation{Automation and Systems Research Institute, Seoul National University, Seoul 08826, South Korea}
\affiliation{Institute of Computer Technology, Seoul National University, Seoul 08826, South Korea}
\affiliation{Institute of Applied Physics, Seoul National University, Seoul 08826, South Korea}
\thanks{Corresponding authors: \href{jbang@yonsei.ac.kr}{jbang@yonsei.ac.kr} and \href{taehyun@snu.ac.kr}{taehyun@snu.ac.kr}}

\date{\today}

\begin{abstract}

Combining variational quantum process with classical neural learning offers a flexible route to improve ground-state estimation. We establish an integrated framework that connects neural transformations, measurement statistics, and physical stability through three requirements: self-contained training, polynomial resource scaling, and variational consistency. Its constructive realization, termed \emph{unitary variational quantum-neural hybrid eigensolver}~(U-VQNHE), couples a variational quantum circuit to a neural phase function through norm-preserving post-processing. The learned transformation is evaluated from measurement records, preserves the exact variational bound, and admits range-independent concentration guarantees for independent finite-shot evaluations. Complementing this construction, we characterize the statistical and representational constraints of amplitude reweighting: sampled-support mismatch can destabilize empirical normalization, while polynomial-range reweighting exactly matches only an exponentially small fraction of Haar-random target distributions from a fixed ansatz, and exact matching for structured ansatz--target pairs can require exponentially large dynamic range under specified near-tensorizability and mismatch conditions. Finite-shot simulations on Ising and disordered XYZ spin models demonstrate improved energy accuracy over the underlying variational quantum eigensolver and greater robustness than the amplitude-reweighting baselines. Together, these results provide a principled foundation for quantum--neural eigensolvers in which physical consistency, expressive capacity, and measurement cost are treated as a single design problem.

\end{abstract}

\maketitle

\section{Introduction}

Variational quantum eigensolver~(VQE)~\cite{Peruzzo2014, Tilly2022} is a paradigmatic near-term algorithm for estimating ground-state energies in chemistry and condensed-matter physics~\cite{McArdle2020, Lim2024, AlBalushi2024}. Its appeal is not only practical---a shallow parametrized quantum circuit (PQC) executed on noisy intermediate-scale quantum (NISQ) hardware~\cite{Preskill2018}---but also physical. Namely, for any normalized trial state, the Rayleigh--Ritz variational principle guarantees that the exact energy expectation value is an upper bound on the true ground-state energy~\cite{Macdonald1933}. This bound anchors optimization to a physical state, while finite-shot energy estimates require statistical error bars. Preserving both this variational interpretation and control of measurement uncertainty is essential when enlarging the variational architecture.

As variational quantum algorithms (VQAs) matured into a broad framework for optimization and quantum machine learning~\cite{Cerezo2021, Farhi2014}, it became natural to hybridize with classical learning systems. One line of work replaces parts of classical architectures with PQCs to exploit quantum representations~\cite{DallaireDemers2018, Lloyd2018, Cong2019, Bausch2020}. A complementary line places neural networks as a post-processing (or pre-processing) module around a variational circuit. Here, the neural encoders can learn problem-adapted embeddings and warm-starts~\cite{PerezSalinas2020, Schuld2019, Grant2019, Mari2020, Miao2023}, while the learned post-processing can mitigate readout bias, denoise observables, and reweight samples~\cite{Bravyi2021, Nation2021, Lowe2021}. These developments point toward a broader opportunity: to integrate quantum state preparation and classical neural learning within a common variational framework. Realizing this opportunity requires understanding how the learned transformation, its measurement estimator, and the available resources jointly determine physical consistency and useful expressivity.

A growing class of quantum--neural algorithms augments circuit-prepared states with learnable diagonal transformations evaluated through classical processing of measurement records. A representative example is the variational quantum-neural hybrid eigensolver (VQNHE), which uses a neural network to reweight computational-basis amplitudes and evaluates the resulting energy through a normalized estimator~\cite{Zhang2022}. Related developments incorporate variational Hamiltonian basis transformations~\cite{Zhang2023}, ancillary post-selection~\cite{Zhang2025}, molecular wavefunction representations~\cite{Li2025}, and separate treatment of amplitude and sign learning~\cite{Ren2025}. More broadly, structured non-unitary correlators and classically optimized diagonal transformations provide other ways to extend the states accessible from variational circuits~\cite{Mazzola2019,Benfenati2021,Gunlycke2024}, which we briefly cover in Appendix~\ref{app:other_nonunitary}.

In this work, we bring the diagonal reweighting components of VQNHE and its variants into a common framework, termed diagonal non-unitary post-processing (DNP). This formulation makes their shared structure explicit: a learnable diagonal transformation modifies the circuit amplitudes, while observables are evaluated through reweighting and normalization. It provides a common setting for analyzing representational capacity, measurement cost, and finite-shot stability under stated assumptions. Placing DNP alongside norm-preserving phase learning then allows us to examine how the choice of classical transformation governs expressivity and variational consistency.

The central question is whether a quantum--neural architecture can combine (i) self-contained training without prior knowledge of the ground state, (ii) polynomial quantum and classical resources, and (iii) a variationally consistent objective with controlled finite-shot evaluation. Diagonal amplitude reweighting exposes a fundamental tension among these requirements: its energy functional is a normalized ratio whose numerator and denominator are estimated from distinct measurement ensembles. Incomplete sampling can therefore affect the two parts asymmetrically, allowing neural optimization to amplify errors in empirical normalization. Resolving this tension requires treating the structure of the learned transformation and the statistics of its estimator as a common design problem.

Our theoretical analysis connects this statistical tension to the expressive capacity of the hybrid representation. For unconstrained amplitude reweighting, we identify sampled-support mismatch as a mechanism for unbounded empirical objectives and quantify the exponential cost of enforcing support inclusion in a uniform-sampling setting. Bounding the neural outputs can stabilize estimation, but also restricts the target distributions accessible to the post-processing layer. We prove a generic noncoverage result for polynomial-range reweighting of a fixed ansatz distribution toward Haar-random targets, and an exponential range obstruction for geometrically local constant-depth ansatz circuits under near-tensorizability and persistent block-wise mismatch conditions. These results characterize concrete limits of amplitude learning and provide structural criteria for designing the neural interface.

We realize these design principles in the unitary variational quantum-neural hybrid eigensolver~(U-VQNHE), organized around norm preservation by construction. A quantum circuit supplies the Born probabilities and initial phases, while a classical neural network learns a complementary phase landscape. The associated diagonal unitary is evaluated virtually through post-processing, so the learned transformation need not be compiled into the state-preparation circuit. Its exact energy obeys the Rayleigh--Ritz bound, and independent finite-shot evaluations admit concentration guarantees without a learned amplitude-range factor or a stochastic normalization denominator. Accessing the real and imaginary contributions requires at most twice the measurement settings of the underlying circuit-transformation scheme. Numerical experiments on Ising and disordered XYZ models demonstrate how this allocation of quantum and classical roles yields practical energy improvements with controlled estimation.

Our contributions can be summarized as follows:
\begin{itemize}
    \item We establish an integrated framework for diagonal quantum--neural hybridization, connecting self-contained training, measurement resources, and variational consistency to the structure of the learned transformation.
    \item We characterize statistical and representational constraints on amplitude reweighting through sampled-support analysis, a Haar noncoverage theorem, and a constant-depth ansatz obstruction under explicit near-tensorizability and mismatch assumptions.
    \item We present U-VQNHE as a norm-preserving quantum--neural architecture, establish exact variational consistency and concentration bounds for independent finite-shot evaluations, and demonstrate improved accuracy and robustness in spin-model benchmarks.
\end{itemize}

\section{Diagonal non-unitary post-processing on VQE}

In the VQE framework, a parameterized unitary $U(\boldsymbol{\theta})$ prepares a variational state
\begin{equation}
    \ket{\psi(\boldsymbol{\theta})} = U(\boldsymbol{\theta}) \ket{0}^{\otimes n},
\end{equation}
which yields the energy expectation value 
$E(\boldsymbol{\theta}) = \bra{\psi(\boldsymbol{\theta})} H \ket{\psi(\boldsymbol{\theta})}$, minimized with respect to $\boldsymbol{\theta}$.  Here the Hamiltonian is expressed as a weighted sum of Pauli strings,
\begin{equation}
H = \sum_P c_P P,
\label{eq:pauli_hamiltonian_app}
\end{equation}
whose expectation values are measurable directly from the quantum circuits via basis transformation~\cite{Peruzzo2014, Wecker2015}. This setup guarantees, from the Rayleigh-Ritz variational principle, that for any circuit ansatz it holds that
\begin{equation}
\min_{\boldsymbol{\theta}} E(\boldsymbol{\theta}) \geq E_{\text{gs}},
\end{equation}
where $E_{\text{gs}}$ is the ground-state energy of the system. Thus, the goal of VQE is to find the set of parameters $\boldsymbol{\theta}$ that minimizes the expectation value of the Hamiltonian via a classical optimizer.


\subsection{Diagonal non-unitary post-processing}
\label{sec:dnp}
DNP extends this paradigm by introducing a classical post-processing procedure that reshapes the measurement distribution generated by the circuit. Instead of directly evaluating the expectation value of the Hamiltonian on $\ket{\psi(\boldsymbol{\theta})}$, it applies a diagonal reweighting operation represented by a generally non-unitary map
\begin{equation}
    \mathcal{D}_f(\rho)
    = \frac{D_f \rho D_f^\dagger}{\Tr[D_f \rho D_f^\dagger]},
    \qquad
    D_f = \sum_{x\in\Omega} f(x)\,\ket{x}\!\bra{x},
    \label{eq:nonunitary_channel}
\end{equation}
where $f(x)\in\mathbb{R}$ (or $\mathbb{C}$) is a parametrized neural network, acting as a function of the measurement output $x$. Here, $\Omega$ denotes the domain of measurement records used in post-processing: in the standard setting $\Omega=\{0,1\}^n$ and $x=s$ is the computational-basis bit string of the system qubits, whereas in some algorithmic variants $\Omega$ may be extended to include additional measured bits or exclude certain bits. Although it can be generalized to encompass complex numbers, we discuss real-valued neural networks in what follows for simplicity.

The transformation acts as a nonlinear map on the density operator and, for a pure input state $\rho=\ket{\psi(\boldsymbol{\theta})}\!\bra{\psi(\boldsymbol{\theta})}$, produces the normalized, post-processed state
\begin{equation}
    \ket{\tilde{\psi}(\boldsymbol{\theta})}_f
    = \frac{D_f \ket{\psi(\boldsymbol{\theta})}}
           {\| D_f \ket{\psi(\boldsymbol{\theta})} \|},
    \label{eq:post_state}
\end{equation}
where $\| D_f \ket{\psi(\boldsymbol{\theta})} \| = \sqrt{\bra{\psi(\boldsymbol{\theta})}D_f^\dagger D_f\ket{\psi(\boldsymbol{\theta})}}$ ensures proper normalization. The corresponding variational energy functional becomes
\begin{equation}
    E_f(\boldsymbol{\theta})
    = \frac{\bra{\psi(\boldsymbol{\theta})} D_f^\dagger H D_f \ket{\psi(\boldsymbol{\theta})}}
           {\bra{\psi(\boldsymbol{\theta})} D_f^\dagger D_f \ket{\psi(\boldsymbol{\theta})}}.
    \label{eq:Ef_general}
\end{equation}
The neural network takes this energy functional as its loss function and repeatedly optimizes the parameters to minimize the loss.

This expression generalizes the standard VQE energy functional, in which the observable expectation value is a linear functional of the prepared quantum state. The diagonal map $D_f$ redefines the variational objective as a nonlinear ratio of reweighted expectation values, effectively acting as a data-dependent filter on the probability amplitudes of $\ket{\psi(\boldsymbol{\theta})}$ in the computational basis. It modifies the underlying Born distribution for each Pauli term in the computational basis, while preserving normalization through the denominator in Eq.~\eqref{eq:Ef_general}. This construction recovers VQNHE~\cite{Zhang2022}, which uses a single classical neural network $f(s): \{0,1\}^n \!\to\! \mathbb{R}$ or $\mathbb{C}$, trained on top of the quantum circuit to minimize $E_f(\boldsymbol{\theta})$. Joint training of the quantum and neural parameters is possible, but finite-shot fluctuations in the energy objective can make optimization unstable. We compare joint and sequential training in Appendix~\ref{app:neuralnetwork}. Expanding the choice of measurement record $x\in\Omega$ and the corresponding estimators yields a broader representation encompassing some algorithmic variants (Appendix~\ref{app:extensions}), all sharing the same diagonal non-unitary structure.

Consequently, Eq.~\eqref{eq:Ef_general} can be regarded as a general variational formulation in which the energy functional depends not only on the unitary circuit parameters $\boldsymbol{\theta}$ but also on the classical parameters defining a nonlinear transformation of the sampling distribution. This abstraction places the aforementioned algorithms within a unified theoretical framework and clarifies that such methods correspond to optimizing an expectation value with respect to a post-processed quantum state $\ket{\tilde{\psi}(\boldsymbol{\theta})}_f$ defined by a diagonal non-unitary map. This formulation highlights the broader applicability of parametrized quantum circuits when combined with diagonal channels of the form~\eqref{eq:nonunitary_channel}, which may improve representational accuracy and convergence efficiency of hybrid quantum--classical algorithms at the cost of modified normalization and sampling complexity, as analyzed in the following sections.

\subsection{General expectation-value evaluation via circuit transformations}

Following Ref.~\cite{Zhang2022}, the neural network training objective in Eq.~\eqref{eq:Ef_general} reduces to estimating normalized contributions of the form
\begin{equation}
    \frac{\bra{\psi(\boldsymbol{\theta})} D_f^\dagger P D_f \ket{\psi(\boldsymbol{\theta})}}
         {\bra{\psi(\boldsymbol{\theta})} D_f^\dagger D_f \ket{\psi(\boldsymbol{\theta})}}
    \quad \text{for each Pauli string } P.
    \label{eq:normalized_P_term}
\end{equation}

The expectation value entering the diagonal post-processing functional can be evaluated using a circuit-transformation scheme. The basic idea follows the original formulation of VQNHE~\cite{Zhang2022}: for each Pauli term of the Hamiltonian, the measurement basis is rotated so that the corresponding operator is effectively diagonalized, and the resulting bit strings are reweighted by the diagonal function $f$ during classical post-processing.

Concretely, for each Pauli term $P$, a short, fixed-depth transformation is applied to align the measurement basis with the eigenbasis of $P$. This involves adding a local entangling operation that coherently encodes the relevant Pauli parity onto a single reference qubit, followed by appropriate single-qubit rotations ($H$ or $R_X(\pi/2)$) that map all Pauli factors to the $Z$ basis. After this step, the circuit is measured in the computational basis, yielding an outcome record $x\in\Omega$ (typically $x=s\in\{0,1\}^n$). Conceptually, this converts the action of the operator $P$ into a pairwise-diagonal form, where each observed bit string $s$ is naturally associated with its partner $\tilde{s}_P$ obtained by flipping the bits corresponding to the non-$Z$ components of $P$. This allows for efficient evaluation of the energy estimator. Refer to Appendix~\ref{app:vqnhe} for details and notations used.


\subsection{Criteria for stable and variationally safe training}
\label{sec:requirements}

Before quantifying resource requirements, we make explicit which properties a neural post-processing module must preserve in order to be a meaningful extension of a variational quantum algorithm. The following desiderata formalize what it means for a post-processing layer to be ``scalable'' and ``physically interpretable,'' and they also clarify precisely where DNP fails under finite sampling.

\begin{definition}[Desiderata for data-driven neural post-processing]
\label{def:desiderata}
We say that a neural post-processing protocol is stable and variationally safe if it satisfies all of the following:
\begin{enumerate}
    \item {\bf Self-contained training (no prior knowledge).} The neural network is trained solely from measurement outcomes on the quantum device during the variational loop, without access to the exact ground state, its energy, or large precomputed training sets.
    \item {\bf Polynomial resource scaling.} The number of distinct quantum circuits, the number of shots per circuit, and the classical overhead (including neural-network size and training cost) scale at most polynomially with the number of qubits $n$.
    \item {\bf Variational consistency.} The reported energy values remain consistent with the Rayleigh--Ritz principle: in the ideal noiseless limit, the objective corresponds to the expectation value of $H$ on a normalized quantum state, and therefore cannot undershoot the true ground-state energy $E_{\mathrm{gs}}$.
\end{enumerate}
\end{definition}

\begin{remark}[Why variational consistency becomes nontrivial for DNP]
For any exactly normalized post-processed state $\ket{\psi}$ (including the formal DNP state in Eq.~\eqref{eq:post_state}), the variational bound $\bra{\psi} H \ket{\psi} \ge E_{\mathrm{gs}}$ holds automatically. The subtlety is that practical DNP implementations do not evaluate Eq.~\eqref{eq:Ef_general} exactly: they estimate its numerator and denominator from independent finite-shot ensembles and optimize the resulting empirical ratio $\hat{E}_f$. This empirical ratio need not coincide with the expectation value of any normalized quantum state and can therefore violate the variational bound unless one expends exponentially many samples.
\end{remark}

The DNP schemes are, at first sight, compatible with the first two desiderata: the network can be trained on-device without supervision, and the diagonal structure keeps the additional classical overhead polynomial. Our central message is that the third desideratum is the bottleneck: in DNP, the normalization itself becomes a sampling-limited quantity. The rest of the paper makes this statement quantitative by analyzing (i) when the empirical normalization can fail catastrophically and (ii) how stabilizing constraints on the neural network introduce an unavoidable tradeoff between expressiveness and sampling cost.

\section{Exponential resource requirement for DNP}
\label{sec:exp_resource_dnp}

As established in Eq.~\eqref{eq:Ef_general}, DNP transforms the standard variational objective into a normalized ratio,
$E_f(\boldsymbol{\theta}) = \sum_P c_P \frac{N_f(P)}{Z_f}$, where
\begin{equation}
\begin{aligned}
N_f(P) &= \bra{\psi(\boldsymbol{\theta})} D_f^\dagger P D_f \ket{\psi(\boldsymbol{\theta})}, \\
Z_f &= \bra{\psi(\boldsymbol{\theta})} D_f^\dagger D_f \ket{\psi(\boldsymbol{\theta})}
    = \mathbb{E}_{s\sim p_\psi}\!\left[\,|f(s)|^2\,\right].
\label{eq:estimators}
\end{aligned}
\end{equation}
Both energy evaluation and gradient estimation therefore hinge on accurately estimating the normalization $Z_f$.
When $D_f$ is non-unitary, the associated weights $w(s)=|f(s)|^2$ distort the underlying Born distribution $p_\psi(s)$, often concentrating it on a few rare outcomes as training proceeds.
This can induce statistical instabilities in normalized observables and cause divergence during neural-network training.
In the following, we show that preventing such divergence requires exponentially many measurements.

\subsection{Shot-based formulation}
In the following sections, we discuss the DNP formulation in the context of VQNHE~\cite{Zhang2022} as a representative example. In shot-based implementations, the numerator and denominator of Eq.~\eqref{eq:Ef_general} are obtained from distinct measurement ensembles. The denominator,
\begin{equation}
    Z_f = \mathbb{E}_{s\sim p_\psi}\!\left[\,|f(s)|^2\,\right],
\end{equation}
is evaluated by sampling the ansatz circuit $U(\boldsymbol{\theta})$ directly in the computational basis, while the numerator requires additional measurement circuits. For a given Pauli string $P$, the corresponding expectation value is estimated from outcomes $\{s_i\}\!\sim p_P$ (the distribution induced by the measurement-augmented circuit for term $P$). Accordingly, we write the probability distribution of the ansatz $p_\psi$ as $p_I$.

With $M_P$ measurements for each Pauli term and $M_I$ measurements for the ansatz, the finite-shot estimators are
\begin{equation}
\begin{aligned}
    \widehat{N}_f(P) &= \frac{1}{M_P}\!\sum_{i=1}^{M_P}
        \sigma_P(s_i)\, f(s'_i)\,f(\tilde{s}'_{i,P}), \\
    \widehat{Z}_f &= \frac{1}{M_I}\!\sum_{i=1}^{M_I} f(s_i)^2,
    \label{eq:empirical_estimators_star}
\end{aligned}
\end{equation}
where $s_i$ denotes the sampled bit string. Here, $s'_i$ is obtained by setting the star-qubit bit $s_{i\star}$ to zero, and $\tilde{s}'_{i,P}$ is the paired bit string defined by the Pauli-induced mapping for term $P$ (see Appendix~\ref{app:vqnhe} for details).
The sign factor $\sigma_P(s_i) = (-1)^{s_{i\star}}$ encodes the parity of the star-qubit outcome $s_{i\star}$ and corresponds to the measured eigenvalue $\pm1$ of $P$. In this measurement setting, the general parity rule reduces to this single-bit sign.

\subsection{Divergence of unconstrained DNP under a sub-exponential number of measurements}

The stability of the DNP estimator depends crucially on two factors: (i) the degree of overlap between the sampling supports of the numerator and denominator in Eq.~\eqref{eq:empirical_estimators_star}, and (ii) the constraints imposed on the neural network that defines the weights. In this section, we show that if no explicit regularization or range constraint is applied to the neural-network outputs---beyond the definitional requirement that $f(s)\ge 0$ for all $s$, since each $f(s)$ acts as a nonnegative weight on a probability component---then exponentially many measurements are required to guarantee arbitrary accuracy. Under any sub-exponential measurement budget, the estimator is generally vulnerable to divergent behavior, often driving the neural network toward extreme values.

To address these issues, we introduce notation for the sampled supports of the numerator and denominator. Let $B_a$ denote the set of bit strings actually observed from the ansatz circuit, and let $B_{m,P}$ denote the set of bit strings observed from the measurement circuit associated with term $P$. The total set of configurations that can contribute to the numerator is
\begin{equation}
    B_M \;=\; \bigcup_P \left\{\,s',\,\tilde{s}'_P : s'\!\in B_{m,P}\,\right\},
    \label{eq:BM_redef}
\end{equation}
where $\tilde{s}'_P$ denotes the partner configuration paired with $s'$ by the Pauli term $P$. For the empirical ratio estimator to be stably normalized, it is necessary that every configuration appearing in the numerator is also sampled; that is,
\begin{equation}
    B_M \subseteq B_a.
    \label{eq:subset_condition}
\end{equation}

\begin{proposition}[Support mismatch yields an unbounded empirical objective]
\label{prop:unbounded_empirical}
Fix a finite-shot dataset used to evaluate the estimators in Eq.~\eqref{eq:empirical_estimators_star}, and let $B_a$ and $B_M$ denote the sampled supports of the denominator and numerator, respectively. If the support-inclusion condition~\eqref{eq:subset_condition} is violated, i.e., $B_M \not\subseteq B_a$, then the empirical DNP objective $\hat{E}_f=\sum_P c_P \hat{N}_f(P)/\hat{Z}_f$ is, in general, not bounded from below over nonnegative weight functions $f$. More precisely, whenever there exists a configuration $s_\star\in B_M\setminus B_a$ that appears in at least one numerator contribution with a negative effective coefficient, one can choose a sequence of nonnegative functions $\{f_k\}_{k\ge 1}$ such that $\hat{E}_{f_k}\to -\infty$ while $\hat{Z}_{f_k}$ remains unchanged.
\end{proposition}

\begin{proof}---By assumption, there exists a term $P_\star$ and a sample outcome in the $P_\star$-measurement ensemble $B_{m,P_\star}$ for which either $s' = s_\star$ or $\tilde{s'}_{P_\star}=s_\star$ appears in the numerator with an overall negative prefactor. Freeze all neural outputs except at $s_\star$, and define $f_k(s_\star)=k$ with $k\to\infty$. Because $s_\star \notin B_a$, the denominator estimator $\hat{Z}_f$ (an empirical average over $B_a$) is independent of $k$. In contrast, the numerator estimator for $P_\star$ contains a contribution proportional to $k$ with negative sign, so $\hat{N}_{f_k}(P_\star)$ decreases without bound as $k\to\infty$, and therefore so does $\hat{E}_{f_k}$.
\end{proof}

Proposition~\ref{prop:unbounded_empirical} shows that, under finite sampling, the DNP loss landscape can acquire directions in which the empirical objective decreases without any normalization penalty. This is a structural instability of the ratio estimator: it disappears only when the denominator samples all configurations that can influence the numerator, i.e., when~\eqref{eq:subset_condition} holds. If this subset condition is violated, some configurations appear only in the numerator and never in the denominator, so their contributions are not properly balanced in the normalized ratio $\widehat{E}_f = \sum_P c_P\widehat{N}_f(P) / \widehat{Z}_f$.

To see the consequences more explicitly, consider a configuration $s^\ast\in B_M\setminus B_a$ that arises in the numerator but is never observed in the ansatz samples. Because the denominator $\widehat{Z}_f$ is computed only over $B_a$, there is no term $|f(s^\ast)|^2$ to counterbalance the corresponding numerator contribution $c_P\,f(s^\ast)\,f(\tilde{s}^\ast_P)\,\sigma_P(s^\ast)$. As a result, this contribution can grow arbitrarily large in magnitude without being compensated by the normalization. During neural-network training, this imbalance can induce a feedback loop: the loss gradient encourages the mismatching terms to increase whenever doing so lowers the estimated energy, whereas the denominator remains insensitive to this change as long as $s^\ast$ is not sampled. Consequently, $f(s)$ can be driven toward extreme values on such unseen configurations, pushing $\widehat{E}_f$ toward divergence even though all $f(s)$ remain nonnegative and numerically finite.

To quantify how costly it is to avoid this scenario, let $N_M = |B_M|$ denote the number of distinct numerator configurations, which can be as large as $2^n$. Determining the exact number of ansatz shots required to guarantee the support-inclusion event $B_M \subseteq B_a$ is nontrivial in general, because $B_M$ itself depends on the measurement circuits and the finite-shot data. Nevertheless, the scaling is already visible in the most favorable case where the ansatz distribution is uniform over all $2^n$ computational-basis strings. Conditioned on a fixed set $B_M$ with $N_M=\abs{B_M}$, collecting all configurations in $B_M$ by repeated i.i.d. uniform sampling is precisely a coupon-collector problem~\cite{Neal2008}. In this setting, the expected number of ansatz shots required to observe every configuration in $B_M$ at least once is
\begin{equation}
    \mathbb{E}[M_I] = 2^n H_{N_M},
    \label{eq:coupon_collector}
\end{equation}
where $H_{N_M}=\sum_{i=1}^{N_M}\frac{1}{i}=\Theta(\log N_M)$ is the $N_M$-th harmonic number.

\begin{proposition}[Coupon-collector scaling for support inclusion]
\label{prop:coupon_collector}
Assume that ansatz measurements are i.i.d. and uniform on $\{0,1\}^n$, and condition on a fixed numerator support set $B_M \subseteq \{0,1\}^n$ with $N_M=\abs{B_M}$. Let $M_I$ denote the (random) number of ansatz shots required until every configuration in $B_M$ has appeared at least once in the ansatz samples (equivalently, until $B_M \subseteq B_a$ holds). Then, $\mathbb{E}[M_I]=2^n H_{N_M}$ as in Eq.~\eqref{eq:coupon_collector}. Moreover, for any $\delta\in(0,1)$, if $M_I^{(\delta)}\ge 2^n\log\big(N_M/\delta\big)$ ansatz shots are taken, then the support inclusion event holds with probability at least $1-\delta$. In particular, if $N_M$ grows exponentially with $n$, then enforcing $B_M \subseteq B_a$ with non-negligible success probability requires exponentially many ansatz measurements.
\end{proposition}

\begin{proof}---For the expectation, reveal the desired configurations one by one. After $k$ distinct elements of $B_M$ have been observed, the probability that the next uniform ansatz shot produces a new element of $B_M$ is $(N_M-k)/2^n$. Hence, the expected waiting time to increase the number of collected elements from $k$ to $k+1$ is $2^n/(N_M-k)$. Summing over $k=0,1,\dots,N_M-1$ gives
\begin{equation}
\mathbb{E}[M_I]=\sum_{k=0}^{N_M-1}\frac{2^n}{N_M-k}=2^n\sum_{j=1}^{N_M}\frac{1}{j}=2^n H_{N_M}.
\end{equation}
For the high-probability bound, fix any $s \in B_M$. The probability that $s$ is never observed in $M$ uniform shots is $(1-2^{-n})^{M}\le \exp(-M/2^n)$. By the union bound over the $N_M$ configurations in $B_M$,
\begin{equation}
\Pr(B_M \not\subseteq B_a)\le N_M\exp(-M/2^n).
\end{equation}
Choosing $M=M_I^{(\delta)} \ge 2^n\log(N_M/\delta)$ makes the right-hand side at most $\delta$, proving $\Pr(B_M\subseteq B_a)\ge 1-\delta$.
\end{proof}

Therefore, under any sub-exponential ansatz-shot budget, one should expect $B_a$ to miss some configurations that can contribute to the numerator, so the condition~\eqref{eq:subset_condition} typically fails. Combined with Proposition~\ref{prop:unbounded_empirical}, this explains why the empirical DNP objective can acquire spurious ``downhill'' directions in the polynomial-shot regime, which the neural network may exploit during training.

This effect can be viewed as a fundamental resource--stability tradeoff: a non-unitary map $D_f$ effectively reweights the Born distribution, but each additional degree of freedom in $f(s)$ introduces a potential instability unless the corresponding region of the outcome space is sufficiently sampled. In the limit where $M_I$ scales exponentially, every $s\in B_M$ is represented, and the estimator converges to its well-defined normalized limit. In contrast, with a polynomial number of measurements, the overlap $B_M\cap B_a$ remains partial, the effective normalization denominator $\widehat{Z}_f$ underestimates the true norm, and the post-processed energy collapses toward $-\infty$.

\begin{figure}

\begin{minipage}[t]{0.88\columnwidth}
\centering
\hspace*{-1.2em}
\includegraphics[width=\textwidth]{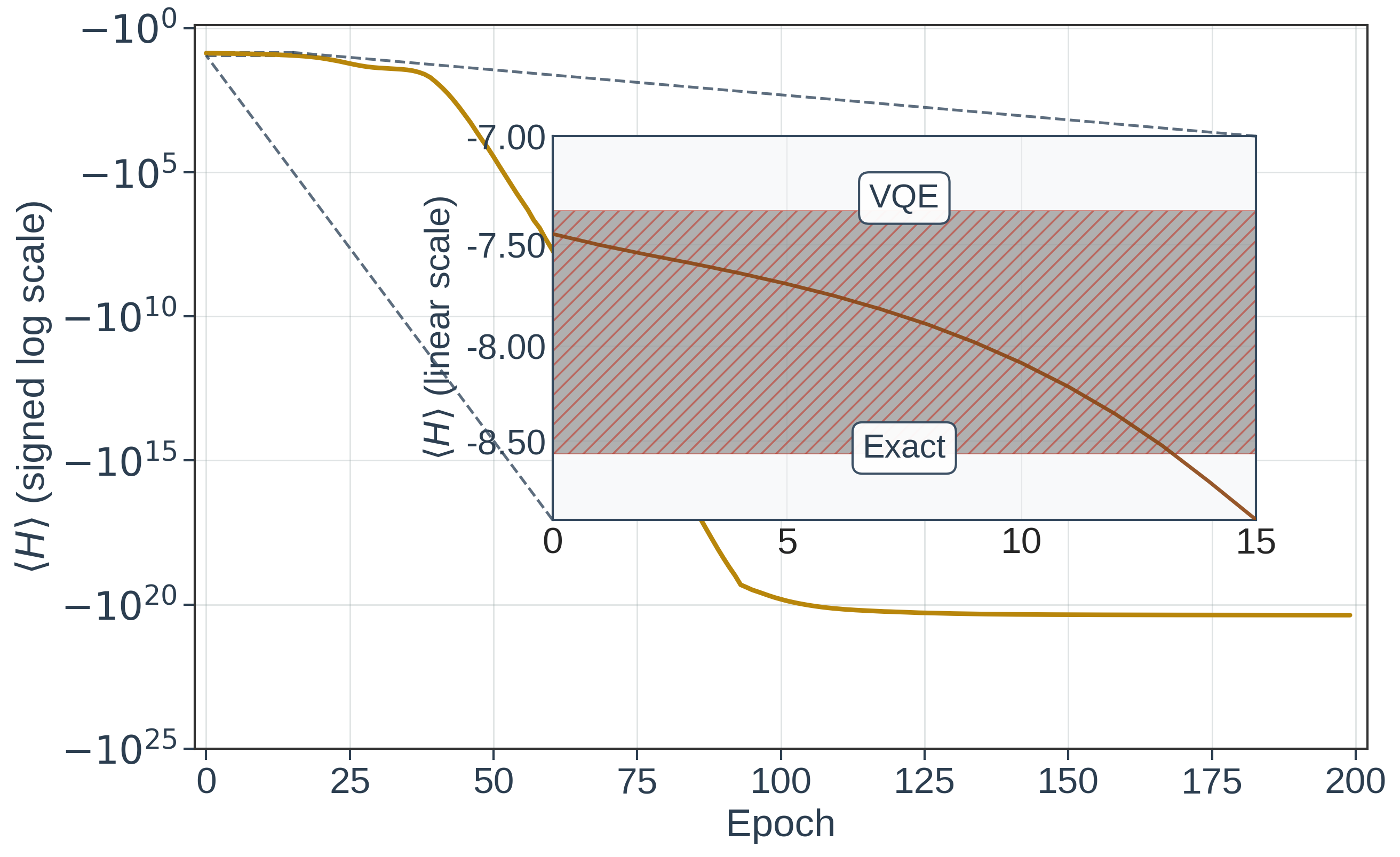}
\textbf{(a)}
\end{minipage}

\vspace{0.5em}

\begin{minipage}[t]{0.82\columnwidth}
\centering
\includegraphics[width=\textwidth]{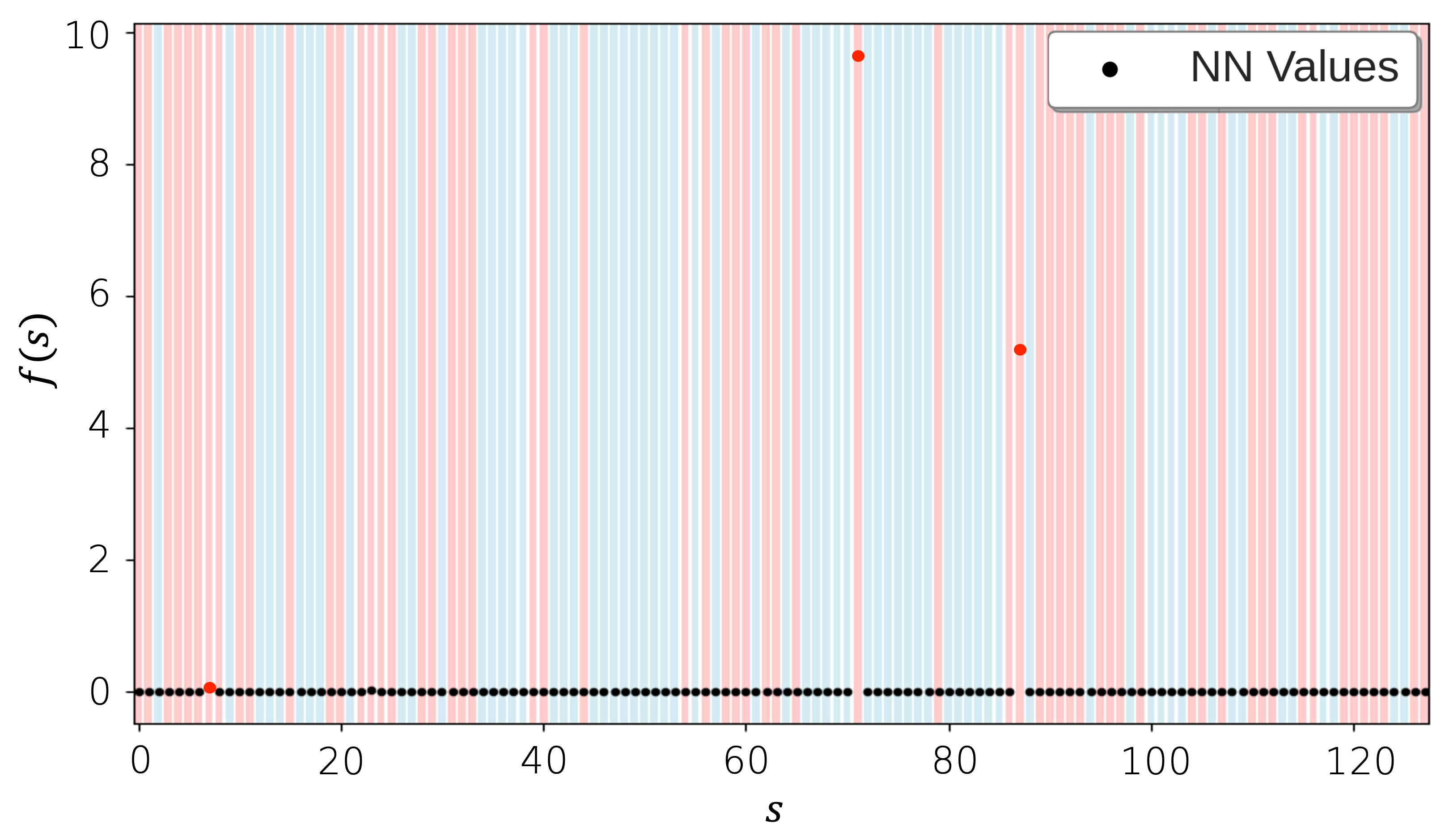}
\textbf{(b)}
\end{minipage}

\caption{\justifying
VQNHE implementation for a 7-site TFIM using 7 qubits.
(a) Training dynamics of the neural network within VQNHE. The vertical axis shows the loss function---the expectation value of the Hamiltonian---plotted on a logarithmic scale. Quantum circuit evaluations are performed using the Qiskit sampler with 500 shots per circuit. The inset highlights the region between the lowest vanilla VQE energy and the exact ground-state energy, noting that values below this range are nonphysical.
(b) Final neural-network outputs after 200 training epochs. The horizontal axis indexes bit strings (in decimal), while the vertical axis shows the corresponding network outputs. Blue lines denote measured bit strings, and red lines denote unmeasured ones. Most outputs cluster near $10^{-6}$, whereas a few extreme values (above $10^{-2}$; red) emerge. Since nearly all bit strings are sampled by the quantum circuits, these red points correspond to strings that contribute only to the numerator.
}
\label{fig:7site}
\end{figure}

We demonstrate this phenomenon using a VQNHE implementation of a 7-qubit transverse-field Ising model~(TFIM), shown in Figure~\ref{fig:7site}. Figure~\ref{fig:7site}(a) displays the optimization trajectory under a shot-based sampler with $500$ measurements per circuit. The estimated energy quickly falls below the exact ground-state value (obtained via exact diagonalization) and diverges to nonphysical magnitudes as the neural network amplifies unsupported configurations.
Figure~\ref{fig:7site}(b) visualizes the trained neural function $f(s)$: outputs corresponding to unmeasured bit strings (red) reach extreme values, while those associated with measured configurations (blue) remain bounded. This confirms that the divergence originates from the incomplete overlap between $B_M$ and $B_a$.

\subsection{Dynamic-range obstruction of bounded DNP}
\label{subsec:dnp_range_obstruction}

To mitigate such divergence while maintaining a practical measurement budget, a straightforward strategy is to impose explicit constraints on the output range of the neural network. Such constraints have been introduced, either implicitly or explicitly, to address training instabilities~\cite{Zhang2022,Ren2025}. In earlier formulations, the network outputs were restricted to be strictly positive, i.e., \(f(s)>0\) for all \(s\). Extending this idea, one can bound the outputs within a finite interval \([1/r,r]\), with \(r>1\), since it is their relative ratios---rather than their absolute magnitudes---that influence the post-processing~\cite{Zhang2022}.

Once the neural-network outputs are confined to a bounded interval, the divergence discussed above can be suppressed by choosing a sufficiently small value of \(r\), which limits the worst-case amplification of the estimator. Moreover, Zhang \textit{et al.} derived that the sufficient measurement budget for achieving an energy error bound \(\epsilon\) is
\begin{equation}
N \geq \frac{9r^4}{4\epsilon^2}.
\label{eq:shot_lowerbound}
\end{equation}
This is a sufficient worst-case prescription rather than a necessary lower bound on the number of shots. Allocating a measurement budget satisfying Eq.~\eqref{eq:shot_lowerbound} suffices to meet the statistical error bound, although fewer shots may suffice for particular instances. This form makes the stability--range tradeoff explicit. Increasing $r$ enlarges the allowed reweighting class, but it also increases the number of measurements prescribed by the bound. Hence, for polynomially bounded $r$ and inverse-polynomial precision, a polynomial measurement budget satisfying this prescription can control the finite-shot error.

This finite-shot stabilization, however, addresses only the statistical behavior of the estimator, not the representational power of the post-processing map. Avoiding worst-case divergence, or even attaining a small estimation error for a fixed bounded range, does not ensure that DNP accurately reproduces the true ground state. If $r$ is chosen too small, the allowed reweighting may be insufficient to transform the ansatz Born distribution into the target ground-state distribution.

Thus bounded DNP exhibits a direct tradeoff between statistical stability and expressive power. Once the neural-network outputs are restricted to a finite range, the post-processing layer can no longer arbitrarily reshape the ansatz distribution; it can only apply a bounded positive reweighting to the computational-basis Born probabilities. More fundamentally, we show below that, for any fixed ansatz distribution, polynomial-range DNP can exactly match only an exponentially small Haar fraction of target Born distributions. For geometrically local constant-depth ansatz circuits, we also establish an exponential range requirement for structured targets under the stated assumptions. Such an exponential range makes the sufficient shot prescription in Eq.~\eqref{eq:shot_lowerbound} exponential as well, without establishing that an exponential measurement budget is necessary. This reveals an architectural constraint of non-unitary post-processing, rather than a merely technical instability.

Let the target and ansatz Born distributions be
\begin{equation}
p_n(s)=|\langle s|\phi_0\rangle|^2,
\qquad
q_n(s)=|\langle s|\psi_\theta\rangle|^2 .
\end{equation}
Exact DNP matching requires a positive weight function satisfying
\[
p_n(s)
=
\frac{w(s)q_n(s)}{\sum_t w(t)q_n(t)} ,
\qquad
w(s)>0 .
\]
We define the exact DNP dynamic range as
\begin{equation}
\Gamma^\star(p_n\mid q_n)
=
\inf_w
\left\{
\frac{\max_s w(s)}{\min_s w(s)}
:
p_n(s)
=
\frac{w(s)q_n(s)}{\sum_t w(t)q_n(t)}
\right\}.
\label{eq:dnp_exact_range}
\end{equation}
Although this condition is only necessary for exact transformation, it captures the essential role of the neural network in DNP: the classical post-processing learns the amplitude structure, while the quantum circuit is responsible for representing the phase structure of the target state. If the post-processing function satisfies \(f(s)\in[1/r,r]\), then \(w(s)=|f(s)|^2\) has range at most \(r^4\). Therefore, exact DNP matching requires
\[
r\ge \Gamma^\star(p_n\mid q_n)^{1/4}.
\]

For full-support distributions, the exact DNP range has the closed form
\[
\Gamma^\star(p_n\mid q_n)
=
\frac{\max_s p_n(s)/q_n(s)}
{\min_s p_n(s)/q_n(s)} .
\]
For singular cases where \(q_n\) has zero-probability components, the exact range identity can be extended by introducing an infinitesimal full-support regularization of the ansatz distribution. This regularization extends the scope of the theory, but does not change the narrative. Thus bounded DNP is efficient only when the target distribution differs from the ansatz distribution by a subextensive log-likelihood correction. Equivalently, polynomial-range DNP can explore only a narrow multiplicative neighborhood of the original ansatz distribution. If the exact DNP range requires exponential scaling, then the desired reweighting is not a benign correction to the ansatz statistics, but a distributional transformation that must amplify or suppress some bit strings by exponentially different factors. The following theorem shows that this obstruction is generic for Haar-random target states.

\begin{theorem}[Generic noncoverage of polynomial-range DNP]
\label{thm:haar_range_noncoverage}
Fix any ansatz Born distribution \(q_n\) on \(n\) qubits, and let \(p_\phi\) be the computational-basis Born distribution of a Haar-random target state \(|\phi\rangle\). Then it is exponentially unlikely that \(p_\phi\) can be reached from \(q_n\) using only polynomial DNP range. More precisely, for any \(R_n=\operatorname{poly}(n)\), we have
\begin{equation}
\Pr_\phi
\left[
\Gamma^\star(p_\phi\mid q_n)\le R_n
\right]
\le
2^{-n}+\exp\left[-2^{n/2}/2\right]
\end{equation}
for all sufficiently large \(n\).
\end{theorem}

The proof can be found in Appendix~\ref{app:dnp_range_obstruction}. Theorem~\ref{thm:haar_range_noncoverage} gives the most general form of our no-go result for a fixed ansatz distribution, without assuming a particular structure of the ground state or the Hamiltonian. Here, Haar-random targets are considered to characterize typical behavior when no particular target in Hilbert space is preferred. The theorem shows that neural reweighting with polynomial dynamic range can exactly reproduce the Born distributions of only an exponentially small fraction of the entire Haar ensemble. This limitation holds for any fixed ansatz distribution, regardless of the neural-network architecture or training procedure. Thus, increasing the expressive power of the neural network alone cannot ensure exact distribution matching while its output range remains polynomial. Although physical ground states are not generally Haar-random, this result establishes a general limitation of bounded reweighting in the absence of target-specific structure. 

For a specified ansatz--target pair, the overlap between their Born distributions provides another way to establish a large DNP range requirement. A useful certificate for this is Bhattacharyya coefficient,
\begin{equation}
B(p_n,q_n)
=
\sum_s \sqrt{p_n(s)q_n(s)} .
\end{equation}
The Bhattacharyya coefficient quantifies the classical overlap between two probability distributions. As this overlap decreases, the ansatz and target Born distributions become increasingly separated, and a larger diagonal reweighting is required to transform one into the other. Thus, a small Bhattacharyya coefficient provides a direct certificate that exact DNP matching demands a large dynamic range.

\begin{proposition}[Bhattacharyya certificate for DNP range]
\label{prop:bc_range_certificate}
For any common-support distributions \(p_n\) and \(q_n\),
\begin{equation}
\Gamma^\star(p_n\mid q_n)
\ge
B(p_n,q_n)^{-4}.
\label{eq:bc-certificate}
\end{equation}
\end{proposition}

\begin{proof}
Consider the R\'enyi-\(\alpha\) divergence
\begin{equation}
D_{\alpha}(p_n\Vert q_n)
=
\frac{1}{\alpha-1}
\log
\left(
\sum_s p_n(s)^{\alpha}q_n(s)^{1-\alpha}
\right).
\end{equation}
The R\'enyi-\(1/2\) divergence 
\begin{equation}
D_{1/2}(p_n\Vert q_n)
=
-2\log
\sum_s\sqrt{p_n(s)q_n(s)}
\end{equation}
satisfies
\[
D_{1/2}(p_n\Vert q_n)
=
-2\log B(p_n,q_n),
\]
and the R\'enyi-\(\infty\) divergence
\[
D_\infty(p_n\Vert q_n)
=
\log \max_s \frac{p_n(s)}{q_n(s)}
\]
is directly related to the exact DNP range. For common-support distributions, the exact range identity gives
\[
\log \Gamma^\star(p_n\mid q_n)
=
D_\infty(p_n\Vert q_n)
+
D_\infty(q_n\Vert p_n).
\]
By monotonicity of R\'enyi divergences,
\[
\begin{aligned}
D_\infty(p_n\Vert q_n)
&\ge
D_{1/2}(p_n\Vert q_n), \\
D_\infty(q_n\Vert p_n)
&\ge
D_{1/2}(q_n\Vert p_n).
\end{aligned}
\]
Since
\[
D_{1/2}(p_n\Vert q_n)
=
D_{1/2}(q_n\Vert p_n)
=
-2\log B(p_n,q_n),
\]
we obtain
\[
\log \Gamma^\star(p_n\mid q_n)
\ge
-4\log B(p_n,q_n).
\]
Exponentiating yields
\[
\Gamma^\star(p_n\mid q_n)
\ge
B(p_n,q_n)^{-4}.
\]
\end{proof}

Therefore, if the classical overlap characterized by the Bhattacharyya coefficient is exponentially small, then exact DNP matching requires exponentially large dynamic range.

We apply this certificate to geometrically local constant-depth ansätze. Such circuits have finite light cones when the circuit geometry is locally connected. Therefore, after discarding buffer regions between separated blocks, the ansatz distribution factorizes over the remaining kept blocks. The target distribution, however, need not factorize exactly. The relevant condition is the near-tensorizability of the target state. We state the following theorem in an abstract form for brevity. A fully explicit version, including the proof, is given in Appendix~\ref{app:dnp_range_obstruction}.

\begin{theorem}[Constant-depth obstruction under near-tensorizability]
\label{thm:near_tensorizable_bc}
Consider a geometrically local constant-depth ansatz circuit applied to a product input state, and choose separated kept blocks whose backward light cones are disjoint. Let \(\bar p_n\) and \(\bar q_n\) denote the target and ansatz Born distributions after discarding the buffers. By locality, the buffer-erased ansatz distribution factorizes across the kept blocks.
\[
\bar q_n(x_1,\ldots,x_m)
=
\prod_{i=1}^m q_i(x_i).
\]
Assume that the buffer-erased target distribution is sufficiently near-tensorizable relative to the ansatz mismatch. Concretely, suppose there exists a constant \(\chi\ge 0\) such that
\[
\bar p_n(x_1,\ldots,x_m)
\le
e^{\chi m}
\prod_{i=1}^m p_i(x_i),
\]
and that the remaining block-wise mismatch has positive density:
\[
\frac{1}{m}
\sum_{i=1}^m
D_{1/2}(p_i\Vert q_i)
-
\chi
\ge
c>0.
\]
Then
\[
B(p_n,q_n)
\le
e^{-\Omega(n)} .
\]
Consequently,
\[
\Gamma^\star(p_n\mid q_n)
\ge
e^{\Omega(n)} .
\]
Thus exact DNP matching from such an ansatz requires exponentially large output range.
\end{theorem}

\begin{figure}[!t]
    \centering
    \begin{minipage}[t]{0.48\columnwidth}
        \centering
        \includegraphics[width=\linewidth]{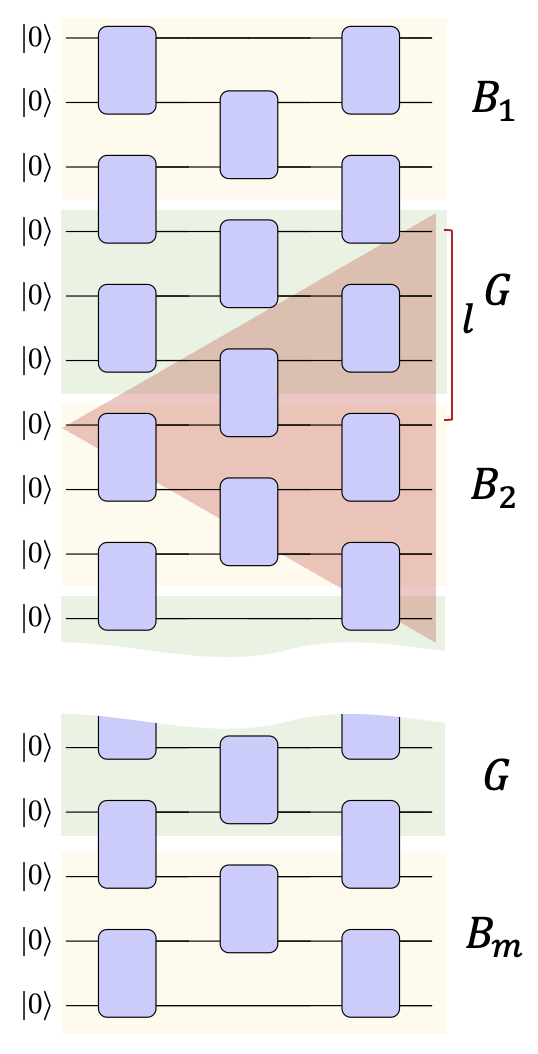}
    \end{minipage}\hfill
    \begin{minipage}[t]{0.48\columnwidth}
        \centering
        \includegraphics[width=\linewidth]{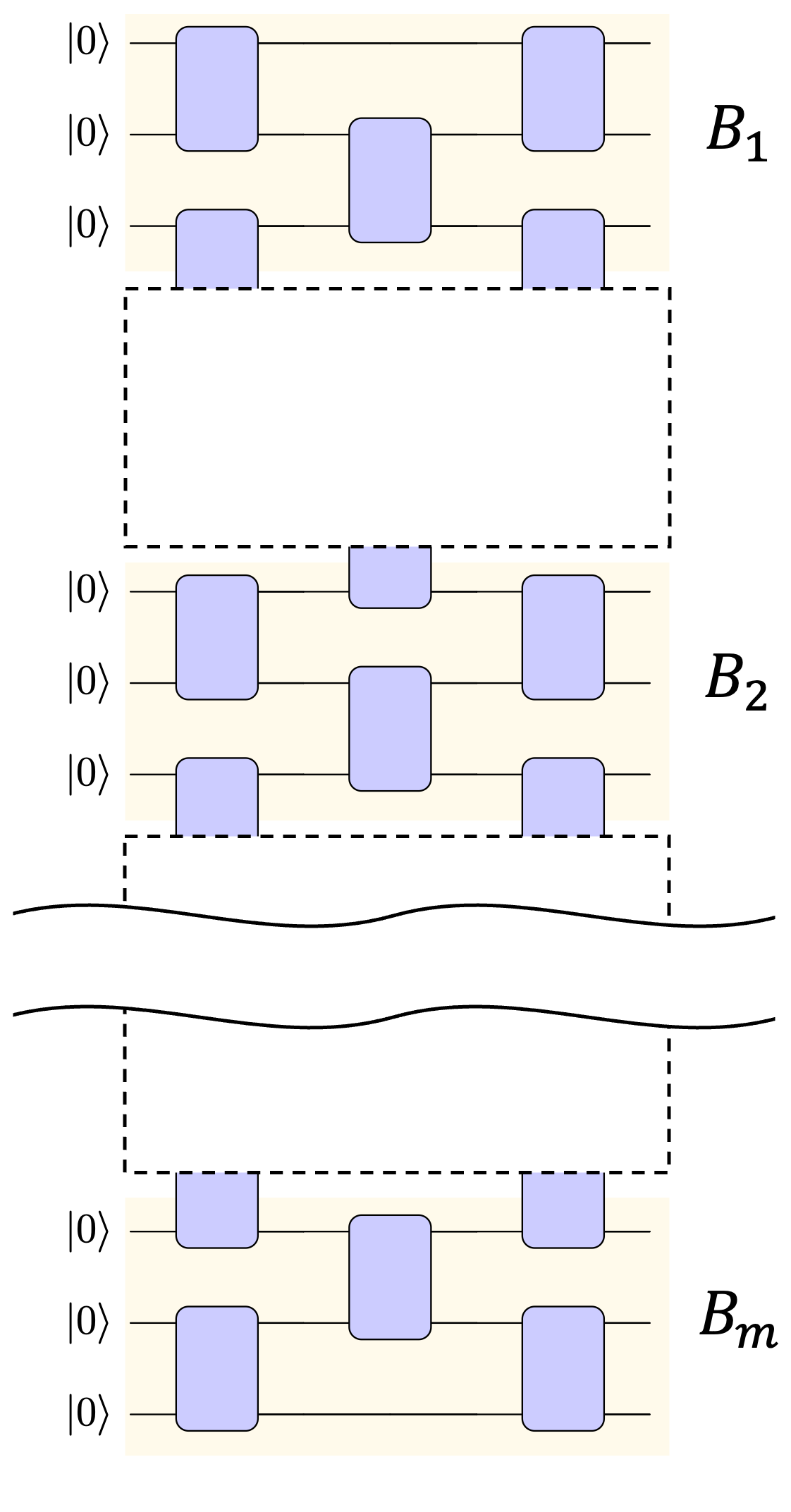}
    \end{minipage}

    \caption{\justifying
    Schematic illustration of the constant-depth Bhattacharyya-overlap obstruction. A geometrically local circuit has a finite backward light cone of width \(\ell=O(d)\), shown by the red wedges. (a) The system is divided into kept blocks \(B_i\), separated by buffer regions whose size exceeds the light-cone width. This separation prevents correlations generated by the constant-depth ansatz from coupling distinct kept blocks. (b) After discarding the buffer regions through a classical stochastic map, the ansatz Born distribution factorizes over the remaining kept blocks. If a positive density of local mismatch between the target and ansatz survives this coarse-graining, the R\'enyi-\(1/2\) divergence accumulates extensively, implying an exponentially small Bhattacharyya coefficient. The formal proof is given in Appendix.}
    \label{fig:finite_depth_lightcone}
\end{figure}

The near-tensorizability condition requires the pointwise dependence between retained blocks to be outweighed by the average local mismatch between the target and the ansatz. It is motivated by systems whose dominant correlations lie within finite regions, as exploited by fragment-based quantum algorithms and shallow VQE approaches that seek to balance accuracy and trainability. Its physical interpretation including chemistry applications and a numerical illustration are discussed in Appendix~\ref{app:dnp_range_obstruction}. 

For target--ansatz pairs satisfying this condition, the theorem and the Bhattacharyya certificate imply that exact distribution matching requires an exponentially large DNP range. This identifies a tradeoff: bounding the range controls weight fluctuations in finite-shot estimation but restricts the changes to the ansatz distribution, whereas allowing a larger range increases the reweighting capability while weakening the range-dependent sampling guarantees. Polynomial-range DNP therefore cannot exactly reproduce the target Born distribution characterized by the theorem.

\section{Practical finite-shot behavior of bounded DNP}\label{sec:practical_scale}

Although the previous section shows that DNP-style algorithms face exponential resource requirements for exact matching in generic or structurally mismatched settings, they may still improve over standard VQE in finite-size practical regimes. In this section, we examine the training behavior of DNP-style algorithms via VQNHE, focusing on neural-network dynamics under practical resource scaling. By varying the number of measurement shots $N$, the dynamic range $r$ of the neural-network outputs, and the number of qubits $n$, we analyze how and why the DNP algorithm exhibits distinct behaviors across these regimes.

\begin{figure*}[t]
\centering

\begin{minipage}[t]{0.32\textwidth}
    \centering
    \includegraphics[width=\textwidth]{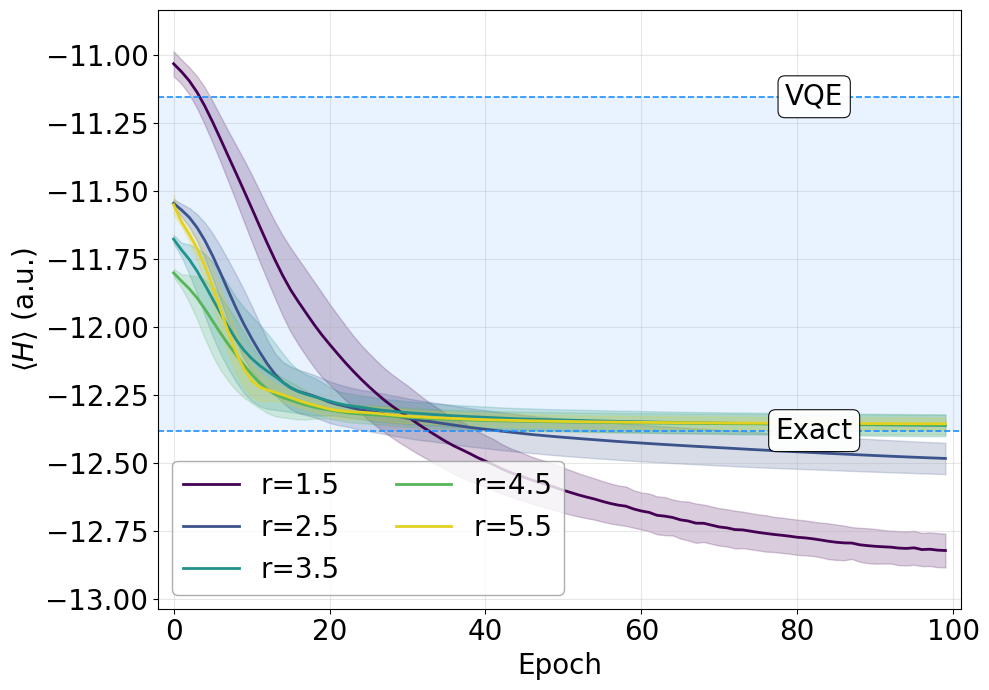}
    \textbf{(a)}
\end{minipage}
\hfill
\begin{minipage}[t]{0.32\textwidth}
    \centering
    \includegraphics[width=\textwidth]{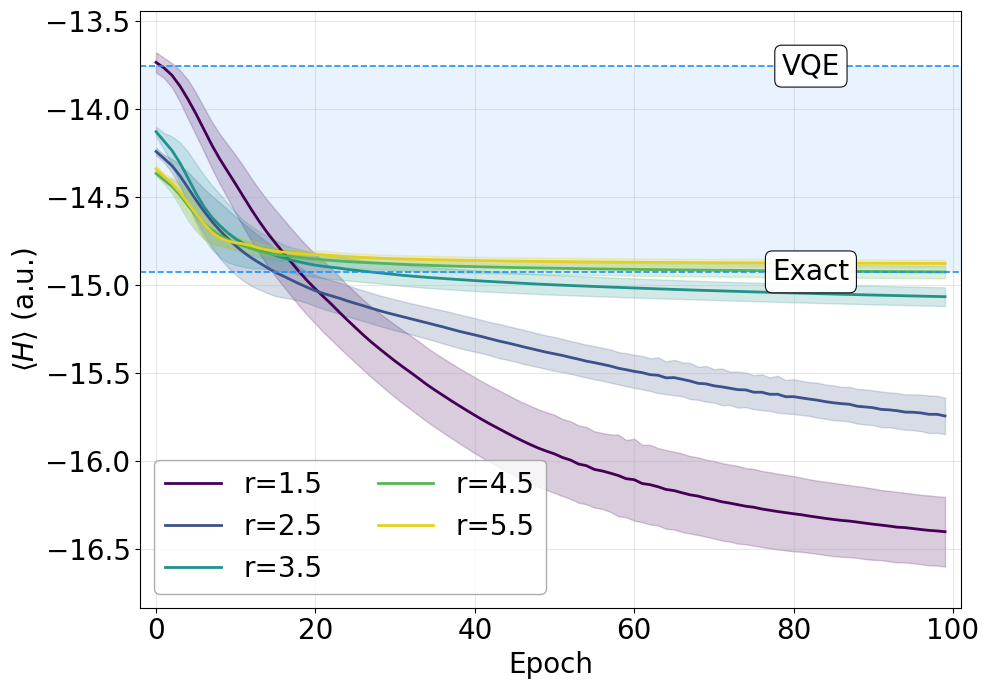}
    \textbf{(b)}
\end{minipage}
\hfill
\begin{minipage}[t]{0.32\textwidth}
    \centering
    \includegraphics[width=\textwidth]{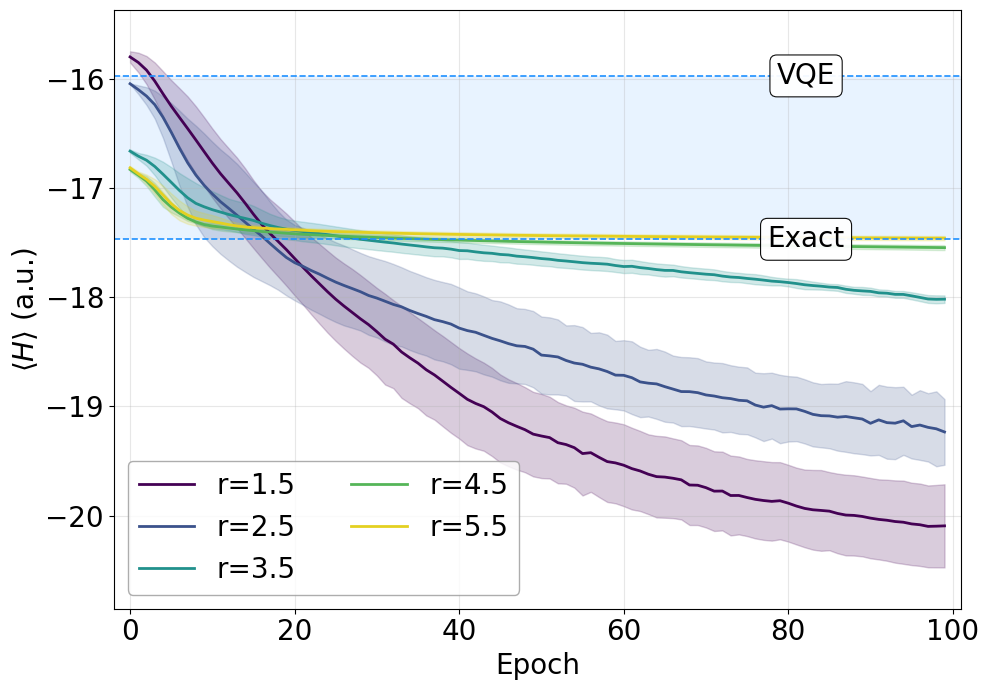}
    \textbf{(c)}
\end{minipage}

\vspace{5mm}

\begin{minipage}[t]{0.37\textwidth}
    \centering
    \includegraphics[width=\textwidth]{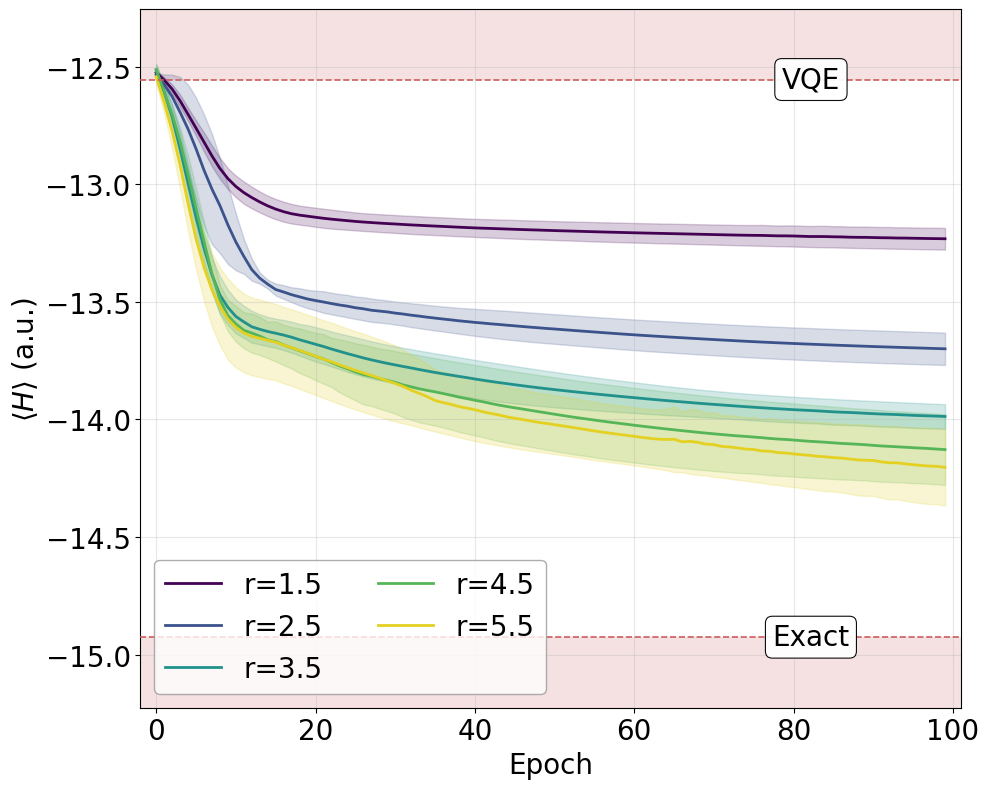}
    \textbf{(d)}
\end{minipage}
\hspace{5mm}
\begin{minipage}[t]{0.37\textwidth}
    \centering
    \includegraphics[width=\textwidth]{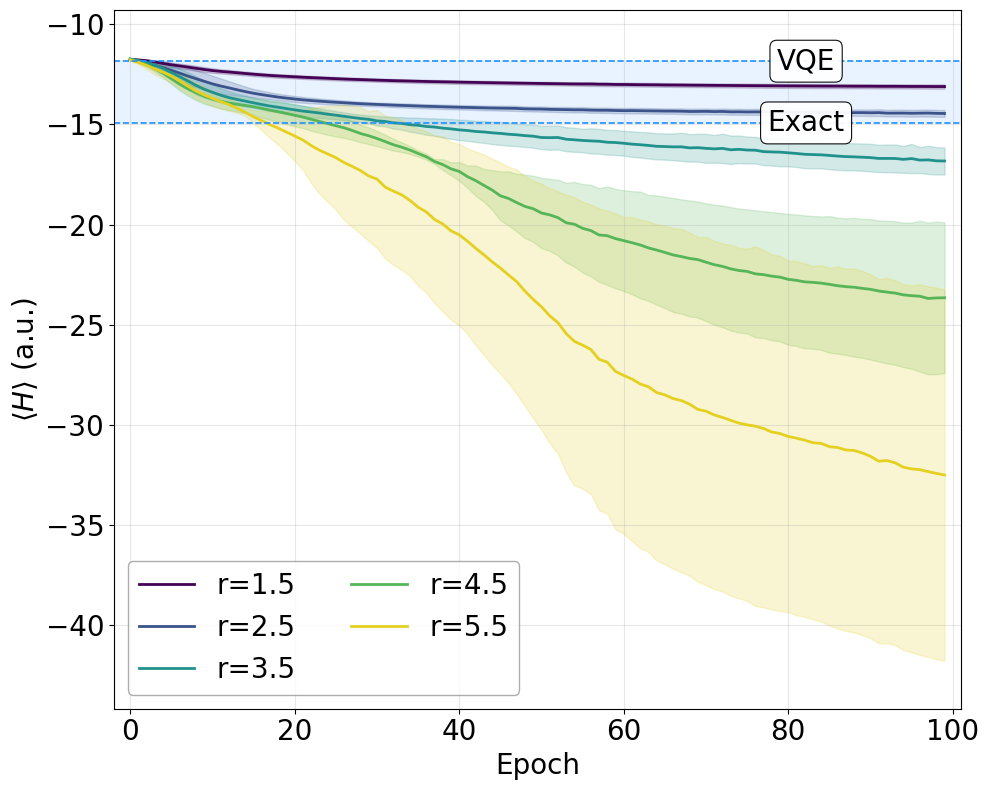}
    \textbf{(e)}
\end{minipage}

\caption{\justifying
Training curves of the neural network in VQNHE with constraints applied, for transverse-field Ising models of various sizes and shot numbers. (a)--(c) show 10-, 12-, and 14-qubit systems, respectively, where the neural-network output is constrained to $f(s)\in[1/r, r]$ for $r$ values from $1.5$ to $5.5$. For each $r$, the number of circuit shots is set to $N = \frac{9r^4}{4\epsilon^2}$ following the sufficient fixed-parameter prescription of Ref.~\cite{Zhang2022}, with $\epsilon=0.05$. The shaded blue regions indicate the range between the VQE baseline and the exact ground-state energy. (d)--(e) correspond to 12-qubit systems trained under fixed-shot conditions: (d) $N=135{,}056$ shots per circuit, chosen according to the sufficient fixed-parameter shot prescription at at $r=3.5$, and (e) $N=10{,}000$ shots per circuit. In these plots, the shaded red or blue regions indicate the area between the VQE baseline and the exact energy.
}
\label{fig:DNP_constraints}
\end{figure*}

\paragraph{Expressiveness under sufficient shot budgets}

Figure~\ref{fig:DNP_constraints}(a)--(c) illustrates the impact of neural-network constraints on the expressiveness of DNP algorithms. Most notably, it shows that across all system sizes, increasing $r$ (and the corresponding number of measurements~\eqref{eq:shot_lowerbound}) pushes the estimated energy closer to the exact ground-state energy. This trend can be interpreted in terms of both the measurement budget and the constraint parameter $r$.

Naturally, as $r$ increases---allowing the neural network to represent a broader range of output values---the network gains the ability to explore a wider solution space beyond what the trained ansatz alone can reach. Although exact target-distribution matching can require an exponentially growing dynamic range under the stated assumptions, in practice $r$ should be chosen to balance expressive power and finite-shot stability within the available measurement budget.

The role of the dynamic range in the expressiveness of DNP can also be seen in a fixed-shot setting, as shown in Figure~\ref{fig:DNP_constraints}(d). Provided that sufficient measurements are taken to prevent the instabilities discussed below, increasing $r$ enhances expressiveness, allowing the estimator to approximate the exact ground-state energy more closely.

\paragraph{Sub-variational behavior under limited shots}

Even with constrained DNP, however, one cannot prevent the estimate from falling below the exact ground-state energy. This phenomenon stems from the need for normalization and inevitable measurement error. Without prior knowledge of the ground state or its energy, the neural network is tasked solely with minimizing the loss function. However, due to the inherent mismatch between the sets of measured bit strings contributing to the numerator $\hat{N}_f$ and denominator $\hat{Z}_f$, the network can artificially lower the loss by assigning extremely large (small) values to bit strings that appear only in the numerator (denominator). Although such behavior does not guide the optimization toward the true ground state, it nonetheless satisfies the formal objective of minimizing the loss function. This manipulation effectively distorts the normalization: adjusting weights on these specific bit strings influences only one of the two estimators---either the numerator or the denominator---thereby breaking their balance.

Figure~\ref{fig:DNP_constraints}(e) illustrates how increasing $r$ can distort the normalization, leading to nonphysical deviations below the true ground-state energy. As discussed, this improper normalization arises when the neural network assigns extreme values to a subset of bit strings---an effect that becomes more pronounced as the dynamic range $r$ grows. In contrast to Figure~\ref{fig:DNP_constraints}(d), where larger $r$ values improve accuracy given sufficient measurement shots, increasing $r$ can deteriorate performance when normalization errors arise due to missing bit-string support in the measured bit strings. The divergence observed in Figure~\ref{fig:7site}(a) is an extreme manifestation of this behavior under unconstrained DNP. In summary, although a larger dynamic range enhances expressiveness and can, in principle, enable closer approximation to the exact ground state, limited measurement resources can cause DNP to produce nonphysical results, potentially leading to critical misinterpretations of the system's behavior.

\paragraph{System-size dependence}

Figure~\ref{fig:final}(a) shows the deviation between the VQNHE results obtained with fixed computational resources and the exact ground-state energies of the TFIM. Since the dynamic range of the neural network is fixed, the expressiveness of DNP is likewise limited, and one might intuitively expect the deviation from the exact solution to grow with system size. However, as shown in the figure, the deviation becomes increasingly negative as the number of qubits grows, indicating not only a systematic deterioration in accuracy, but also that the solution is nonphysical. This worsening trend highlights the difficulty of maintaining reliable finite-shot estimation as the system size increases under fixed resource conditions.

In conclusion, the DNP framework exhibits several noteworthy behaviors during neural-network training. First, when a sufficient number of measurement shots are available---though potentially beyond a practical low-shot regime---a wider output range of the neural network enables a more accurate approximation to the ground-state energy. In contrast, within the polynomial-shot regime, increasing the dynamic range can degrade performance, often producing nonphysical energies below the true ground-state value due to normalization issues. This deviation becomes progressively more severe as the number of qubits increases, reflecting the growing difficulty of maintaining stability under limited measurement resources.

\section{Stable training via unitary post-processing}\label{sec:uvqnhe}

\begin{figure*}[tb]
\centering

\begin{minipage}[t]{0.48\textwidth}
    \centering
    \includegraphics[width=\textwidth]
    {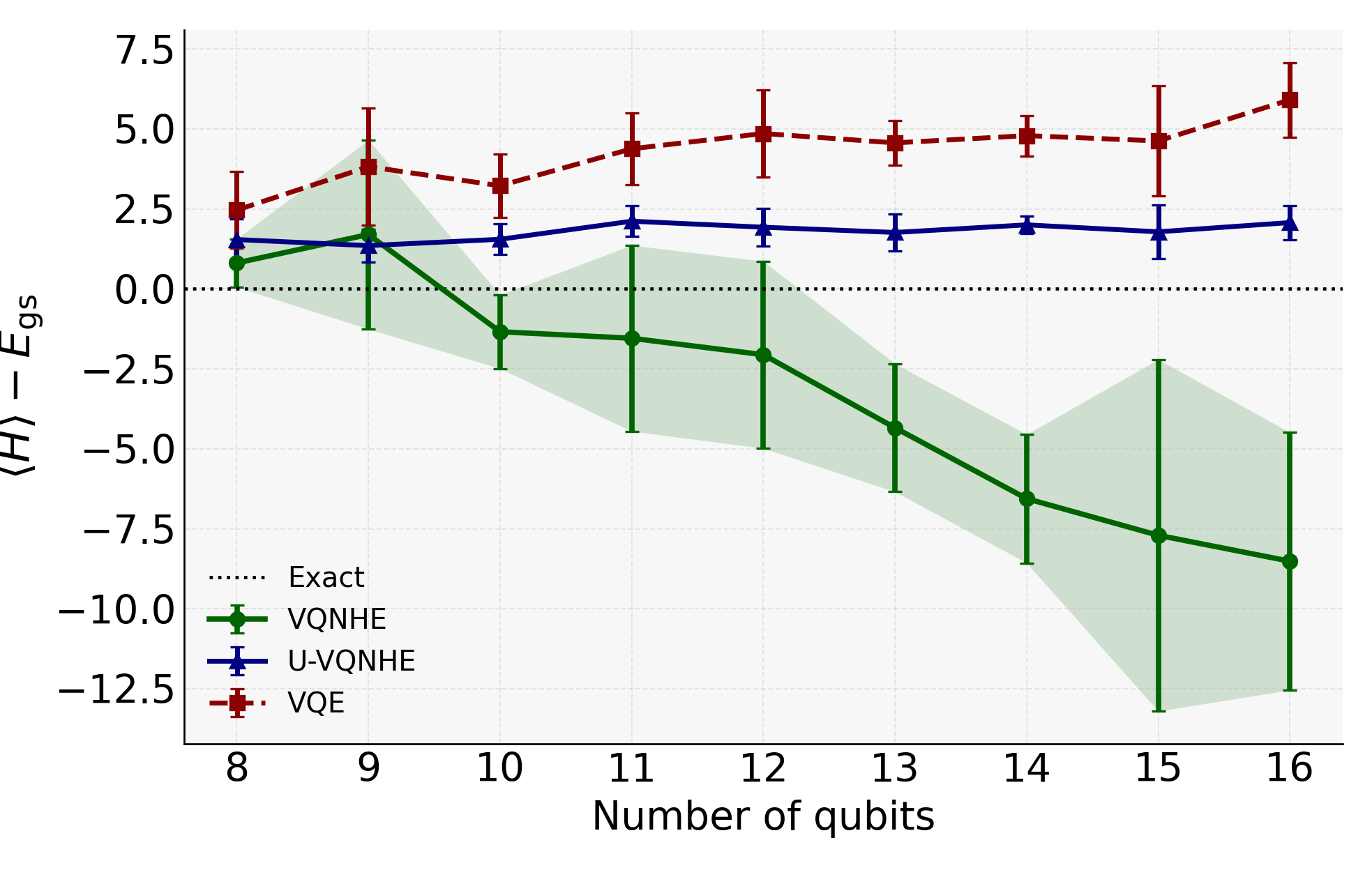}
    \textbf{(a)}
\end{minipage}
\begin{minipage}[t]{0.42\textwidth}
    \centering
    \includegraphics[width=\textwidth]{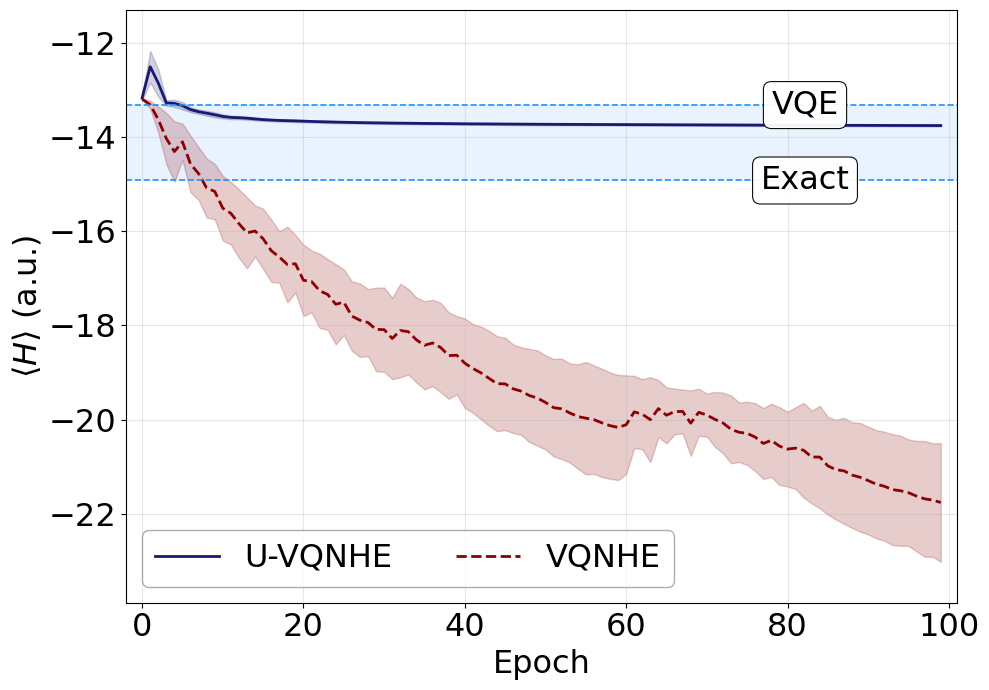}
    \textbf{(b)}
\end{minipage}

\caption{\justifying
Comparison of the training performance of VQNHE and U-VQNHE.
(a) Energy difference from the ground-state energy for TFIM instances with increasing numbers of qubits, using 10{,}000 measurements per circuit. The neural-network output is bounded within $[1/3,3]$, as before. The red dashed line shows the gap between VQE and the ground-state energy. As the number of qubits increases, the deviation between VQNHE and the exact ground-state energy generally worsens.
(b) Results for a 12-site TFIM using a two-layer ansatz and 10{,}000 measurements per circuit. The region between the exact ground-state energy and the VQE result is shown in light blue. The solid blue line shows the result from U-VQNHE, while the dotted red line shows that of VQNHE with constrained DNP (with range parameter $r=3$). VQNHE violates the variational bound, leading to nonphysical solutions. By contrast, U-VQNHE is generated from an exactly normalized post-processed state, so its exact objective remains variational, with deviations controlled by the concentration bound in Appendix~\ref{app:uvqnhe_finite_shot}.
}
\label{fig:final}
\end{figure*}

\begin{figure*}[tb]
\centering

\begin{minipage}[t]{0.32\textwidth}
    \centering
    \includegraphics[width=\textwidth]{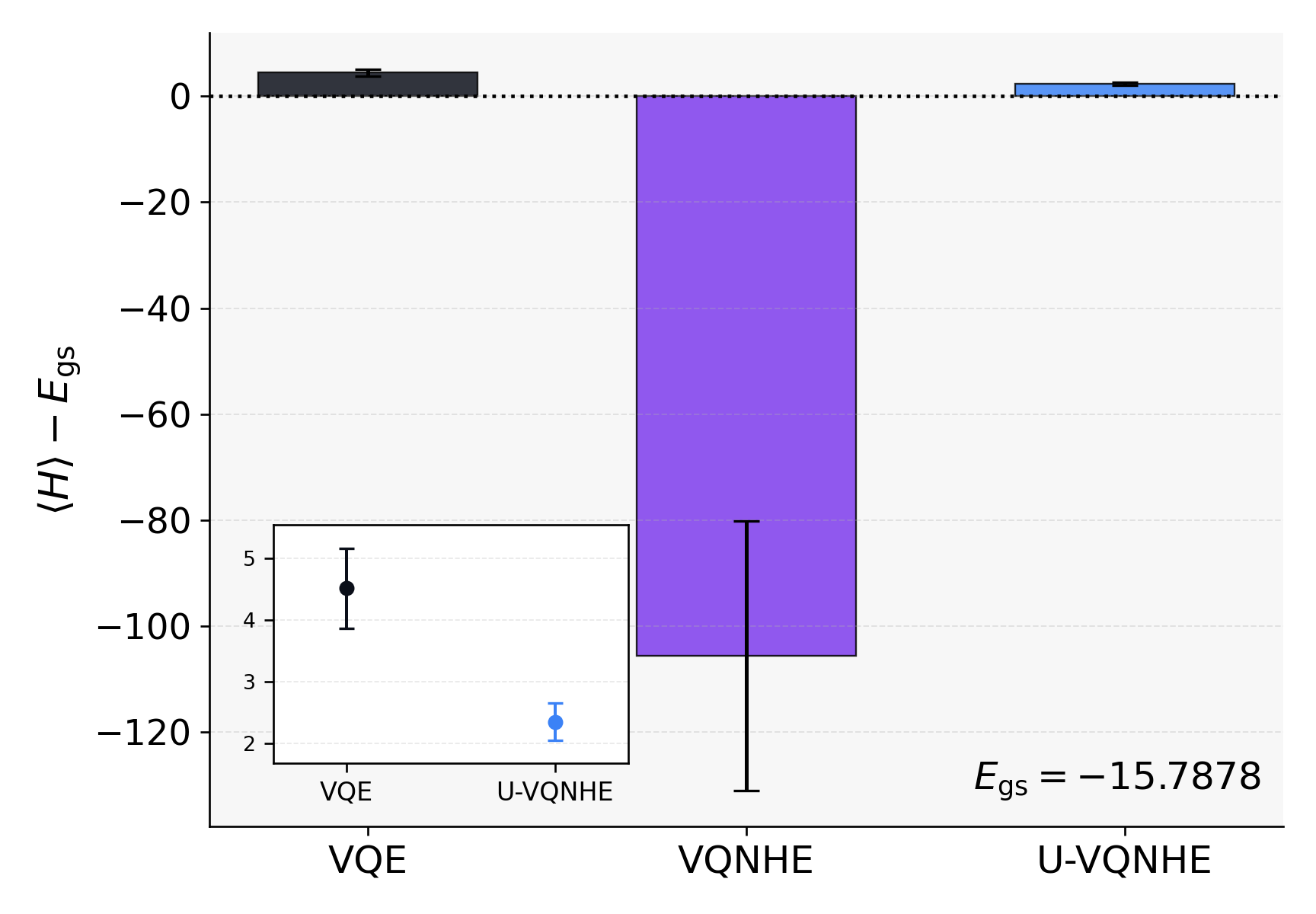}
    \textbf{(a)} Random Ising model
\end{minipage}
\hfill
\begin{minipage}[t]{0.32\textwidth}
    \centering
    \includegraphics[width=\textwidth]{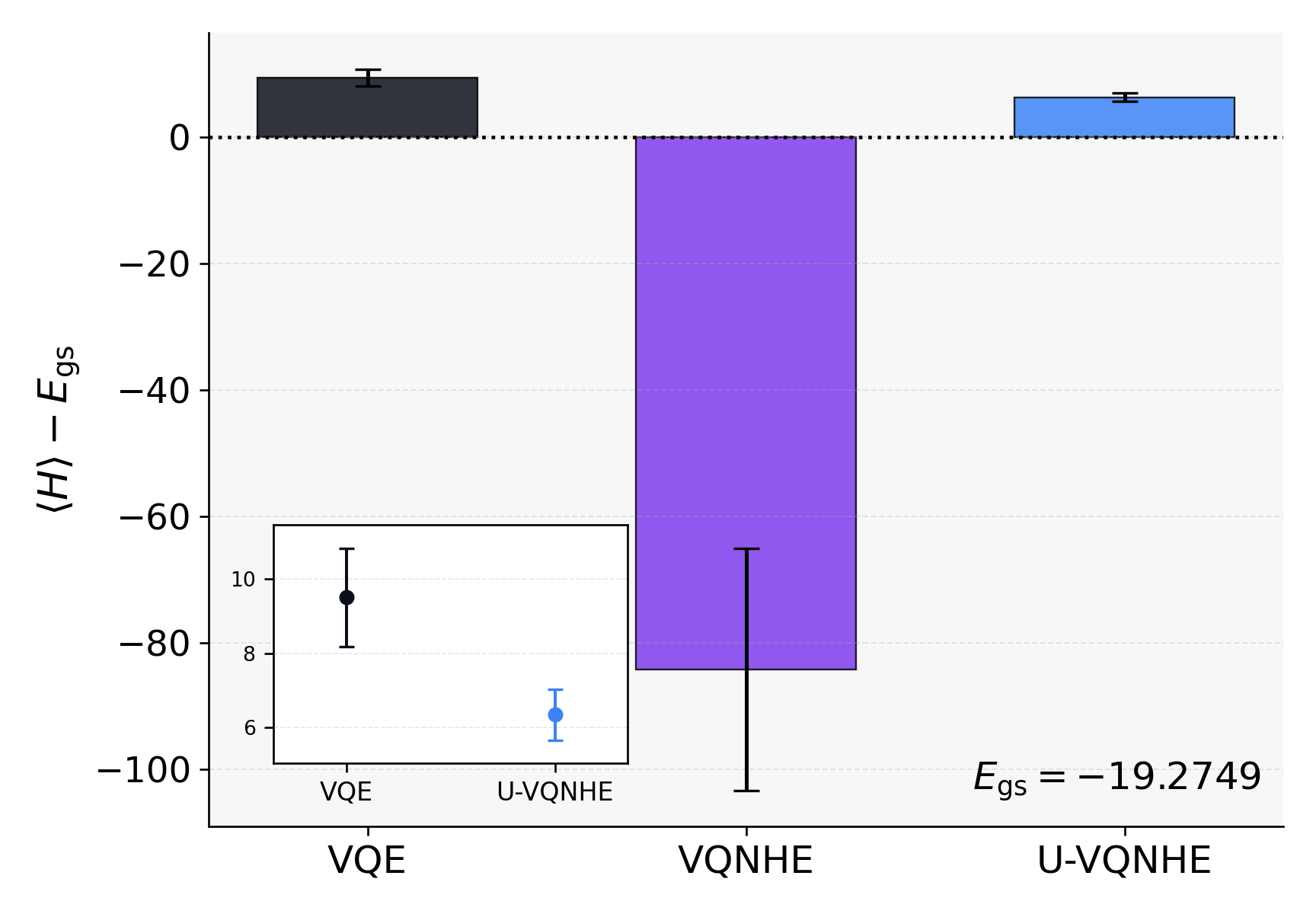}
    \textbf{(b)} Two-dimensional TFIM
\end{minipage}
\hfill
\begin{minipage}[t]{0.32\textwidth}
    \centering
    \includegraphics[width=\textwidth]{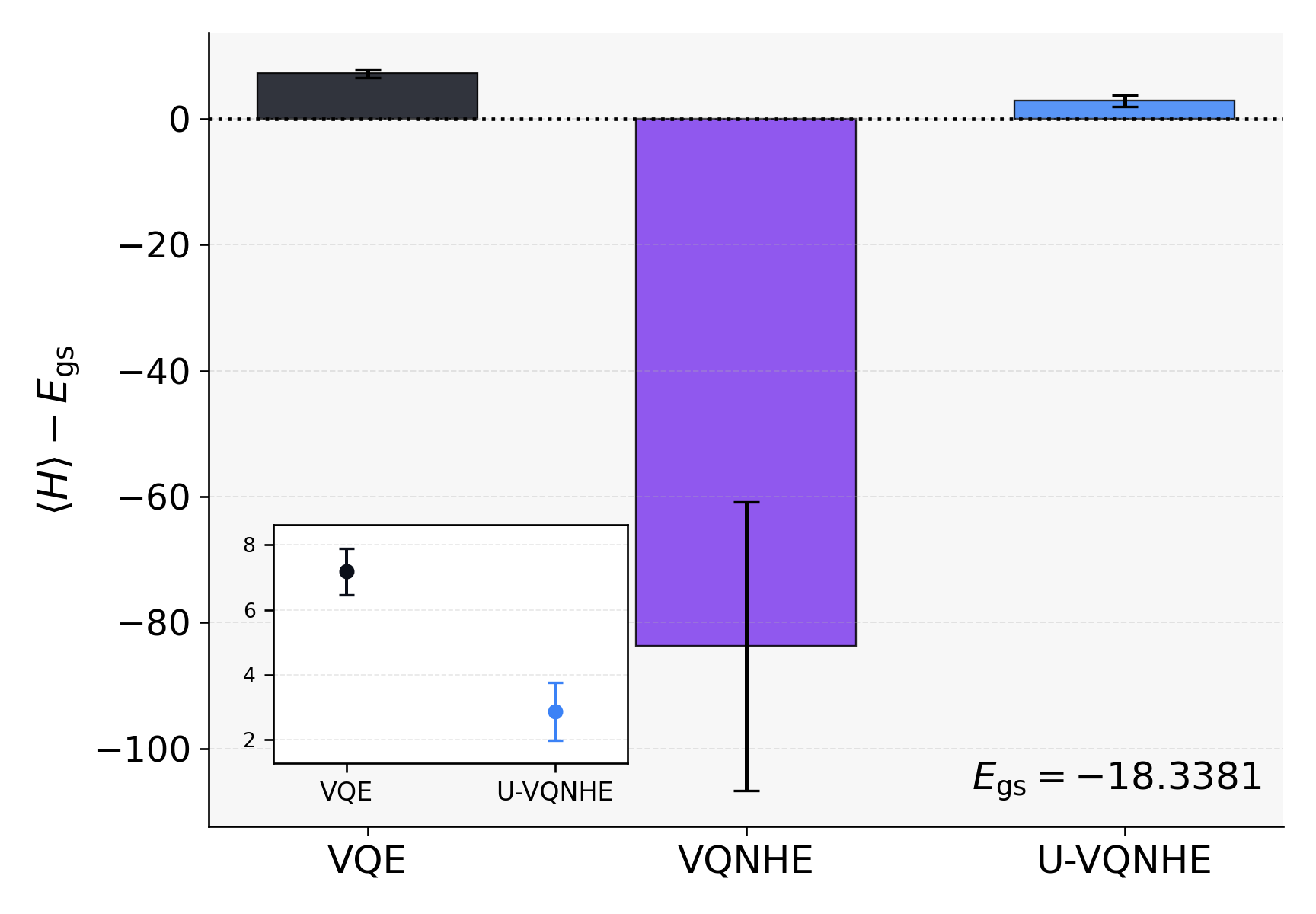}
    \textbf{(c)} Disordered two-dimensional XYZ spin model
\end{minipage}

\caption{\justifying
Energy-estimation performance of VQE, VQNHE, and U-VQNHE applied to different systems. Each panel reports the energy error $\langle H \rangle - E_{\mathrm{gs}}$, taking the exact ground-state energy $E_{\mathrm{gs}}$ as the zero reference. Lower absolute error indicates a more accurate ground-state estimate. Markers denote the mean over five independent repetitions and error bars indicate one standard deviation. The inset provides a magnified comparison of the VQE and U-VQNHE energy errors relative to the exact ground-state energy. All trials used 1024 quantum measurements per circuit and 1000 neural-network training epochs. (a) A 12-qubit one-dimensional random-coefficient Ising Hamiltonian with nearest-neighbor $ZZ$ couplings and random transverse and longitudinal fields. (b) A 12-qubit two-dimensional transverse-field Ising model, which increases connectivity and introduces a less one-dimensional correlation structure. (c) Disordered XYZ spin model on a \(3\times3\) square lattice with open boundary conditions, random anisotropic nearest-neighbor couplings, and random local fields along all three spin directions. All results use finite-shot circuit sampling, and energies are reported relative to the exact ground-state energy obtained by sparse diagonalization.
}
\label{fig:hamiltonian_comparison}
\end{figure*}

Sections~\ref{sec:requirements}--\ref{sec:exp_resource_dnp} show that the explicit normalization intrinsic to DNP is not a benign implementation detail: under finite sampling, normalization turns the DNP objective into a ratio estimator whose value can be dominated by rare configurations. As a consequence, the polynomial measurement budgets can lead to support mismatch between numerator and denominator estimators (Sec.~\ref{sec:exp_resource_dnp}), and the empirical loss landscape may develop nonphysical directions that violate the Rayleigh--Ritz bound. This motivates a simple design principle for quantum-neural hybridization to preserve the normalization by construction. If the post-processing layer is norm-preserving, then (in the absence of additional noise) its objective is automatically the expectation value of $H$ on a valid quantum state, and variational consistency follows immediately.

We implement this principle by replacing the diagonal non-unitary filter $D_f$ with a diagonal unitary transformation. A diagonal map is unitary if and only if each diagonal entry has unit modulus. Hence, instead of learning amplitudes $f(s)\in\mathbb{R}$ as in DNP, we let the neural network output a real-valued phase function $g_\phi(s)\in\mathbb{R}$ and define the diagonal unitary
\begin{equation}
    {
    \ket{\psi_u(\boldsymbol{\theta})}
    = U_g\,\ket{\psi(\boldsymbol{\theta})},
    \quad
    U_g = \sum_{s\in\{0,1\}^{n}} e^{i g_\phi(s)} \ket{s}\!\bra{s}.
    }
    \label{unitarytransformation}
\end{equation}
This phase-only post-processing is compatible with the complex-valued formulation discussed in the supplementary material of Ref.~\cite{Zhang2022}, but crucially it eliminates the need for explicit (and sampling-limited) renormalization.

\begin{theorem}[Variational safety of diagonal unitary post-processing]
\label{thm:unitary_safety}
Let $U_g = \sum_{x\in\Omega} e^{i g(x)}\ket{x}\bra{x}$ be a diagonal unitary defined by an arbitrary real-valued function $g:\Omega \to \mathbb{R}$. For any normalized ansatz state $\ket{\psi(\boldsymbol{\theta})}$, the post-processed state $\ket{\psi_u(\boldsymbol{\theta})}=U_g\ket{\psi(\boldsymbol{\theta})}$ is normalized and satisfies the Rayleigh--Ritz bound
\begin{equation}
    \bra{\psi_u(\boldsymbol{\theta})} H \ket{\psi_u(\boldsymbol{\theta})}\ge E_{\mathrm{gs}}.
\end{equation}
In particular, any exact evaluation of the U-VQNHE objective is variationally consistent by construction.
\end{theorem}

\begin{proof}---The unitarity of $U_g$ implies $\norm{\psi_u(\boldsymbol{\theta})}=\norm{\psi(\boldsymbol{\theta})}=1$. The Rayleigh--Ritz principle then gives $\bra{\psi_u}H\ket{\psi_u}\ge E_{\mathrm{gs}}$ for any normalized state $\ket{\psi_u}$.
\end{proof}

Operationally, U-VQNHE inherits the circuit-transformation estimator of VQNHE: the diagonal unitary $U_g$ is applied only virtually through classical post-processing of measurement outcomes. Because the learned phase factors introduce complex interference, estimating $\bra{\psi_u}P\ket{\psi_u}$ requires access to both real and imaginary parts; this can be achieved by at most doubling the set of measurement circuits used in VQNHE. Importantly, the estimator remains linear in shot-averaged quantities and does not involve an explicit normalization ratio. Derived directly from VQNHE, the expectation value from this transformed state is evaluated as
\begin{align}
&\langle \hat{H} \rangle_g \nonumber \\
&= \sum_P c_P \big[ \sum_{s \in B_{m,P}}  (-1)^{s^\star} \operatorname{Re} \left( e^{-ig(s')} e^{ig(\tilde{s'}_P)} \right) p_m(s;P) \nonumber \\
&+ \sum_{s \in B_{m',P}} (-1)^{s^\star} \operatorname{Im} \left( e^{-ig(s')} e^{ig(\tilde{s'}_P)} \right) p_{m'}(s;P) \big].
\label{uvqnhe}
\end{align}
Here, $p_m(s;P)$ is the probability distribution of the measurement circuits of VQNHE, and $p_{m'}(s;P)$ is that of the measurement circuits corresponding to the imaginary parts. The total number of measurement circuits that must be evaluated is at most doubled, but the phase factors introduce no amplitude-range-dependent shot overhead.

U-VQNHE preserves the circuit state's Born probabilities and therefore cannot correct amplitude errors. Its benefit depends on whether relative-phase corrections can lower the energy. No improvement is possible when the energy is insensitive to these phases or when the circuit already has energy-minimizing phases for its distribution. Nevertheless, a large ground-state overlap is not required for an energy improvement. The attainable accuracy remains constrained by the Born distribution supplied by the ansatz, while diagonal unitary post-processing provides additional phase freedom without introducing the amplitude-range-dependent sampling overhead of non-unitary reweighting.

The finite-shot stability of this estimator is analyzed in Appendix~\ref{app:uvqnhe_finite_shot}, where the corresponding variance and concentration propositions are given. In particular, when each measurement component is estimated with \(N\) shots, the finite-shot estimate \(\widehat E_g\) concentrates around the exact U-VQNHE energy \(E_g\) with the usual \(N^{-1/2}\) sampling rate and with a prefactor controlled only by the Hamiltonian coefficient norm \(\|c\|_2\). This marks the difference from typical DNP implementations: since \(U_g\) is unitary, the exact energy \(E_g\) remains variational, satisfying \(E_g\ge E_{\mathrm{gs}}\). The finite-shot estimate \(\widehat E_g\), however, is an ordinary statistical estimator and may fluctuate above or below \(E_g\), and therefore may also fall below \(E_{\mathrm{gs}}\) in a particular finite-shot realization. What U-VQNHE provides is not an unconditional guarantee on the raw sampled value, but a high-probability certificate
\begin{equation}
E_{\mathrm{gs}}
\le
E_g
\le
\widehat E_g
+
2\|c\|_2
\sqrt{\frac{\log(2/\delta)}{N}},
\end{equation}
which is independent of any neural-network amplitude range or stochastic normalization denominator.

Figure~\ref{fig:final} compares U-VQNHE with the constrained non-unitary post-processing baseline. The results illustrate the main advantage of U-VQNHE: the exact objective remains within the Rayleigh--Ritz variational framework because the post-processing map is unitary. Thus the transformed state \(\ket{\psi_u}\) is always a normalized physical state, and its exact energy satisfies \(E_u\ge E_{\mathrm{gs}}\). A finite-shot estimate of this energy can still fluctuate statistically, but the deviation from \(E_u\) is controlled by the range-independent concentration bound in Appendix~\ref{app:uvqnhe_finite_shot}. In contrast, constrained VQNHE can reduce finite-shot instabilities but still relies on non-unitary amplitude reweighting and a normalization ratio, so its finite-shot behavior does not admit the same range-independent certificate.

Figure~\ref{fig:hamiltonian_comparison} compares VQE, VQNHE, and U-VQNHE on more complex benchmark Hamiltonians. In all three cases, U-VQNHE improves the energy estimate relative to vanilla VQE while its exact objective remains protected by the Rayleigh--Ritz variational bound. By contrast, VQNHE produces energies far below the exact ground-state energy, reflecting the loss of variational control caused by its non-unitary normalization. Details of the Hamiltonians and numerical settings are provided in Appendix~\ref{app:benchmark-details}.

U-VQNHE also shows stable optimization when the quantum and neural parameters are trained jointly. In the finite-shot comparison in Appendix~\ref{app:neuralnetwork}, joint U-VQNHE reaches an energy close to its sequential-training result, whereas joint VQNHE does not improve upon the pretrained VQE reference. Across ten independent sequential runs, U-VQNHE gives similar final results from different neural-network initializations, and its gradient norms show neither explosive growth nor a premature collapse to zero. These observations complement the exact Rayleigh--Ritz bound and finite-shot concentration guarantee with evidence of reliable neural-network training in the tested setting.

\section{Discussion and outlook}\label{sec:discussion}

The central message of this work is that quantum--neural eigensolvers should be designed by considering quantum state preparation, classical representation learning, and measurement statistics together. Our framework connects these ingredients through self-contained training, polynomial resource scaling, and variational consistency, and provides both analytical constraints and a constructive architecture. The analysis of amplitude reweighting shows how the learned transformation determines the statistical structure of the objective: explicit normalization introduces a ratio estimator that can be dominated by sparsely sampled configurations. We identify two complementary constraints: (i) sampled-support mismatch can make unconstrained empirical objectives unbounded, with an exponential support-coverage cost in the uniform-sampling setting (Sec.~\ref{sec:exp_resource_dnp}), and (ii) exact target-distribution matching can demand exponentially large reweighting range for Haar-random targets or structured ansatz--target pairs satisfying the stated near-tensorizability and mismatch assumptions (Theorems~\ref{thm:haar_range_noncoverage} and~\ref{thm:near_tensorizable_bc}). Together with the finite-shot experiments in Sec.~\ref{sec:practical_scale}, these results establish why expressive capacity and statistical stability must be assessed jointly.

From a physics perspective, these results give variational consistency an operational meaning by separating the exact objective from its sampled estimate. For a normalized state $\ket{\psi}$, the exact energy $\bra{\psi}H\ket{\psi}$ satisfies the Rayleigh--Ritz bound; finite-shot estimates can fluctuate around this value even in standard VQE. Neural amplitude reweighting introduces an additional mechanism: optimizing an empirical normalization ratio can amplify sampling errors into spurious directions of decreasing loss. Norm preservation removes this mechanism at the level of architecture. In U-VQNHE, every choice of neural parameters defines a normalized state, and the bounded measurement contributions yield explicit concentration guarantees for independent evaluations at fixed parameters (Appendix~\ref{app:uvqnhe_finite_shot}). This connects a physical variational objective to a statistically controlled energy estimate.

U-VQNHE implements this design principle through a complementary allocation of quantum and classical roles. The circuit determines the computational-basis Born probabilities, while the neural network learns relative-phase corrections that modify off-diagonal energy contributions through interference. This diagonal unitary preserves normalization exactly (Theorem~\ref{thm:unitary_safety}) and can be evaluated through classical processing of measurement records. Its expressivity is specifically phase-based. For a fixed circuit state, it cannot correct amplitude errors or change the diagonal contributions to the energy. Nevertheless, energy improvement does not require a large overlap with the ground state. It requires phase corrections that lower the energy within the fixed Born distribution. No improvement is possible when the energy is insensitive to these phases or when the circuit already has energy-minimizing phases for its distribution. For example, if the Hamiltonian is diagonal in the computational basis, its energy depends only on the Born probabilities and is therefore unchanged by U-VQNHE, even when the VQE energy remains above the ground-state energy. Because the identity transformation is included, the original VQE state remains available within the enlarged variational family, although finite-shot training need not attain its optimal energy.

Our Ising and disordered XYZ benchmarks demonstrate cases where additional phase freedom improves energy accuracy over the underlying VQE under finite measurement budgets, with greater robustness than the amplitude-reweighting baselines studied here. For TFIM instances of increasing system size, Appendix~\ref{app:phase_amplitude_analysis} separates the VQE energy error into the portion removable by phase corrections and the residual error at fixed Born probabilities. These results characterize U-VQNHE as a complementary extension of VQE whose benefit depends on the Hamiltonian and the prepared circuit state, combining additional phase expressivity with a variational objective and an estimator free of stochastic normalization.

Several directions follow from this integrated perspective. First, norm preservation provides an organizing principle for broader quantum--neural architectures, including learned shallow unitary transformations and trace-preserving quantum channels whose observables remain efficiently estimable. Second, the joint analysis of expressivity and sampling cost can guide the design of learning interfaces that either avoid explicit normalization or control it through provable concentration guarantees. Finally, U-VQNHE can be combined with adaptive ansatz growth, symmetry constraints, or subspace-expansion techniques to study complementary benefits in larger and more structured Hamiltonians. Extending these ideas requires understanding how representation, optimization, and independent energy validation interact as system size grows. The resulting theory and architecture provide a foundation for pursuing these questions while keeping physical consistency and measurement resources explicit throughout quantum--neural design.

\section*{Acknowledgments}

This work has been supported by the National Research Foundation of Korea (NRF) grant (No. RS-2024-00442855, No. RS-2024-00413957, No. RS-2024-00432214, and No. RS-2025-18362970), and the Institute of Information \& Communications Technology Planning \& Evaluation (IITP) grant (No. RS-2022-II221040), all of which are funded by the Korean government (MSIT). We also thank Luning Zhao for mentoring and helpful discussions.

\section*{Author contributions}

M.K. conceived the project, identified the issues, proposed the algorithm, developed the mathematical framework, and wrote the manuscript. K.P. and K.L. provided guidance on the project and feedback on the manuscript. The simulation code was implemented by M.K.. J.B. and T.K. supervised the project, including research planning and manuscript preparation.

\section*{Code availability}

Python scripts that generate all figures and execute the algorithms are archived at Github,
\url{https://github.com/miinukim/vqnhelib}.

\section*{Use of generative AI}

The authors used OpenAI's ChatGPT for language refinement and for assistance with organizing and formatting research data and manuscript materials. All output was critically reviewed and verified by the authors, who take full responsibility for the manuscript and its reported results.

\appendix

\section{Variational quantum-neural hybrid eigensolver}
\label{app:vqnhe}

VQNHE~\cite{Zhang2022} augments a standard VQE by inserting a diagonal non-unitary post-processing
\begin{equation}
D_f \;=\; \sum_{s\in\{0,1\}^n} f(s)\,\ket{s}\!\bra{s},
\end{equation}
that reweights measurement outcomes of a variational state $\ket{\psi}$. This is a standard setting of DNP where we have $\Omega = \{0,1\}^n$ and $x=s$. The reweighting converts Born probabilities $p_\psi(s)=|\langle s|\psi\rangle|^2$ into $p_f(s)\propto |f(s)|^2 p_\psi(s)$, increasing the effective expressiveness of a shallow ansatz without changing its circuit structure.

Let $\hat H=\sum_P c_P\,P$ with Pauli strings $P$. For each $P$, choose a \emph{star qubit} $q_\star$ within the support of the $X/Y$ factors of $P$. For every other qubit $j$ in the $X/Y$ support of $P$ with $j\neq q_\star$, apply a controlled gate onto $q_\star$ with control $j$: use controlled-$X$ if $P$ has $X$ on $j$ and controlled-$Y$ if $P$ has $Y$ on $j$. This aggregates the $X/Y$-parity of $P$ onto the star qubit. After this aggregation, perform the standard computational-basis diagonalization for $P$: apply $H$ on qubits where $P$ has $X$, $R_X(-\pi/2)$ on qubits where $P$ has $Y$, and measure directly where $P$ has $Z$ or $I$. A final computational-basis readout then yields bit strings $s=(s_\star,\bar{s})$, and the eigenvalue of $P$ is encoded in the star-qubit outcome via $\sigma_P(s)=(-1)^{s_\star}$.

Let \(f:\{0,1\}^n\to\mathbb{R}\) define \(D_f\) as above. We restrict
this derivation to real-valued post-processing, while the detailed construction can be found in~\cite{Zhang2022}. Denote by $p_P(s)$ the shot-normalized probability of $s$ obtained from the circuit used to measure $P$, and by $p_I(s)=p_\psi(s)$ the probability from the bare ansatz circuit, or equivalently from the Pauli term $I$. The chosen diagonalization of $P$ induces a deterministic map $s\mapsto \tilde s_P$ on data bit strings such that it exerts XOR operations on the support of $X$ and $Y$. For example, if $P = XYZ, s = 011$, then $\tilde{s}_P = 101$. Also define a filter $s\mapsto s'$ that sets the star-qubit bit of $s$ to $0$ in $s'$ while leaving all other bits unchanged, or in other words, $s=(s_\star,\bar{s}) \mapsto s' = (0, \bar{s})$.

With these definitions, the energy estimator reads
\begin{equation}
\label{eq:vqnhe-energy}
\big\langle \hat H \big\rangle_f
\;=\;
\frac{\displaystyle \sum_{P} c_P \sum_{s=(s_\star,\bar{s})} f(s')\,f(\tilde{s'}_P)\,\sigma_P(s)\,p_P(s)}
{\displaystyle \sum_{s} |f(s)|^2\,p_I(s)}\,.
\end{equation}
In practice, the sums are implemented as Monte Carlo averages over the corresponding shot ensembles: the numerator averages $c_P\,f(s'_i)\,f(\tilde s'_{P,i})\,\sigma_P(s_i)$ over samples $s_i$ drawn from the $P$-measurement circuits, while the denominator averages $|f(t_j)|^2$ over samples $t_j$ drawn from the ansatz circuit itself, or with respect to the identity term. Equation~\eqref{eq:vqnhe-energy} makes explicit that the quantum hardware supplies unweighted bit strings and all diagonal reweighting is performed in classical post-processing via $D_f$.

\section{Neural post-processing algorithms from the DNP framework}
\label{app:extensions}
In this appendix we detail how several descendants of the variational quantum--neural family operate in practice and how each can be written as an expansion of the DNP map
\[
\mathcal{D}_f(\rho)= \frac{D_f \rho D_f^\dagger}{\Tr[D_f \rho D_f^\dagger]},
\qquad
D_f=\sum_{x\in\Omega} f(x)\ket{x}\!\bra{x},
\]
where $x$ denotes the measurement record used in post-processing and $\Omega$ is its domain. Throughout, the underlying variational state is $\ket{\psi(\boldsymbol{\theta})}=U(\boldsymbol{\theta})\ket{0}^{\otimes n}$ (or, when additional measured registers are present, the corresponding state on the enlarged register), and the post-processed state is
\begin{equation*}
\ket{\tilde{\psi}(\boldsymbol{\theta})}_f
=
\frac{D_f\ket{\psi(\boldsymbol{\theta})}}
{\|D_f\ket{\psi(\boldsymbol{\theta})}\|},
\end{equation*}
with
\begin{equation*}
\|D_f\ket{\psi(\boldsymbol{\theta})}\|
=
\sqrt{\bra{\psi(\boldsymbol{\theta})}D_f^\dagger D_f\ket{\psi(\boldsymbol{\theta})}}.
\end{equation*}

\subsection{Variational post-selection}

Variational post-selection~\cite{Zhang2025} augments a variational state with one or more ancilla qubits. After state preparation, measurement outcomes are \emph{filtered} before forming expectation values, thereby biasing the effective sample distribution toward a target subspace or target statistics while preserving a normalized estimator. In the ground-state setting this filtering is implemented by projecting onto a designated set of accepted ancilla outcomes, while in the thermal-state setting it is implemented via a graded acceptance profile $g\in[0,1]$, i.e., a neural reweighting of ancilla outcomes prior to normalization.

At the density-matrix level, let the joint system--ancilla state be
$|\Psi\rangle=\sum_{i,m}\psi_{im}\,|i\rangle_{\mathrm{sys}}|m\rangle_{\mathrm{anc}}$.
Reweighting the ancilla outcome $m$ by $g(m)$ and tracing out the ancilla yields the reduced system state
\begin{equation}
\rho^{(\mathrm{sys})}_{ij}(g)\ \propto\ \sum_{m}\psi_{im}\,g(m)\,\psi^{*}_{jm}\,,
\label{eq:soft-post-rho}
\end{equation}
which specializes to projector-based post-selection when $g$ is an indicator on the accepted ancilla set.

From the DNP viewpoint, this is a diagonal, non-unitary post-processing on the extended measurement record $x=(s,m)\in\Omega$, with $\Omega=\{0,1\}^n\times\mathcal{M}$, implemented by
\begin{align}
&D_{f_{\mathrm{ext}}}
=\sum_{s,m} f_{\mathrm{ext}}(s,m)\,|s,m\rangle\!\langle s,m|,\\
\nonumber &f_{\mathrm{ext}}(s,m)=\sqrt{g(m)}\,f(s),
\end{align}
followed by the usual DNP renormalization. In other words, variational post-selection remains within DNP by enlarging the outcome domain and absorbing the post-selection (or neural acceptance) into the diagonal weights.

\subsection{pUNN}

Relative to standard DNP, which applies a diagonal weight $f(s)$ to single-register bit strings, paired unitary coupled-cluster with neural networks~(pUNN) enlarges the space on which the diagonal weights act~\cite{Li2025}. In the DNP definition, this corresponds to replacing the outcome domain $\Omega=\{0,1\}^n$ by a product domain $\Omega=\mathcal{K}\times\mathcal{J}$ and applying a diagonal, non-unitary weight function $F(k,j)$ on the joint measurement record $x=(k,j)$. Concretely, pUNN pairs the paired unitary coupled-cluster with double excitations ansatz~\cite{Henderson2015} with an auxiliary circuit of the same size where classical perturbation is given before it is entangled with the system, and lets the diagonal weight depend on both indices, $F:\mathcal{K}\times\mathcal{J}\to\mathbb{R}$ (or $\mathbb{C}$). This lift gives the model enough capacity to encode correlations that a single-index $f(s)$ cannot.

Under this viewpoint, pUNN is simply DNP on the enlarged register, with diagonal operator
\[
D_F=\sum_{(k,j)\in\mathcal{K}\times\mathcal{J}} F(k,j)\,|k,j\rangle\!\langle k,j|,
\]
and the energy is then computed using the same normalized estimator structure as DNP. Operationally, the auxiliary tag $j$ gives the neural post-processing an extra conditioning variable. This allows pUNN to represent reweightings that cannot be captured by a function of the hardware outcome $k$ alone---intuitively, it can learn different reweighting rules for the same $k$ depending on the accompanying tag $j$, thereby encoding additional correlations.

If one averages over the auxiliary tag, the scheme collapses back to an effective single-variable reweighting on $k$, recovering standard DNP in form. Keeping the pair $(k,j)$ explicit is therefore the distinguishing feature of pUNN: it is best viewed as DNP applied to an enlarged outcome space, with the same ratio-of-expectations structure but a richer input domain for the diagonal weights.

\subsection{sVQNHE}

Signed-VQNHE (sVQNHE)~\cite{Ren2025} builds on diagonal neural post-processing, but modifies the hybrid design and training procedure to more explicitly accommodate sign/phase structure. Within our DNP definition, sVQNHE retains the same measurement-outcome domain (typically $\Omega=\{0,1\}^n$) and the same diagonal reweighting layer specified by a neural network $f$ acting on bit strings, while enforcing additional restrictions on the admissible outputs of $f$ (and associated regularization) to improve numerical conditioning and reduce estimator variance. In parallel, sVQNHE assigns the learning of sign/phase information to the quantum circuit by incorporating commuting diagonal gates, which provide a structured and measurement-efficient mechanism for encoding phases beyond what the diagonal neural reweighting alone can capture.

Operationally, sVQNHE employs a layer-wise, phase--amplitude decoupling strategy with a bidirectional feedback loop between the neural network and the PQC. At each iteration, a new block of quantum layers is added: a classically simulatable, non-diagonal layer (used to shape amplitudes) together with a diagonal layer (used to learn phases). The algorithm first performs a forward initialization step, where the amplitude structure learned by the neural operator from the previous iteration is transferred to the newly added simulatable layer via a tractable classical optimization. Then, in the subsequent hybrid optimization step, the neural operator is reset and jointly optimized together with the parameters of the new diagonal layer, following the same variational objective as in basic DNP. Leveraging the commuting nature of the diagonal layers, sVQNHE is designed to reduce measurement overhead for updating the phase-related parameters, while retaining the DNP component as the amplitude-learning module. Furthermore, Ref.~\cite{Ren2025} argues that their layer-wise structure mitigates normalization issues, whereas our analysis shows that normalization remains an intrinsic issue of the DNP structure.

\subsection{VQNHE++}

The three variants above remain \emph{within} the DNP abstraction in the strict sense that they only modify the diagonal weight function (and/or its domain) while preserving the same post-processing map
\(
\mathcal{D}_f(\rho)=D_f\rho D_f^\dagger/\Tr[D_f\rho D_f^\dagger]
\)
and the same normalized ratio estimator. By contrast, VQNHE++~\cite{Zhang2023} should be viewed as a \emph{DNP-based extension}: it retains the DNP layer unchanged, but augments the overall procedure with an additional classical degree of freedom that acts on the Hamiltonian representation.

In the DNP setting we keep the quantum state preparation fixed within a measurement round, collect rotated-basis histograms that diagonalize each Hamiltonian term, and apply a diagonal (generally non-unitary) post-processing $D_f=\sum_{s} f(s)\,|s\rangle\!\langle s|$ (equivalently, $x=s\in\Omega=\{0,1\}^n$). The energy is evaluated as introduced in Eq.~\ref{eq:vqnhe-energy}, where $H=\sum_P c_P P$, $p(t)=|\langle t|\psi\rangle|^2$ is the bare-ansatz distribution, $p_P(s)$ is the distribution from the circuit used to measure $P$, $s \mapsto s'$ is the star-collapse map that sets the star bit to $0$, $\tilde s'_P$ is the companion bit string induced by the chosen diagonalization of $P$, and $\sigma_P(s)=(-1)^{s_\star}$ is the eigenvalue sign read from the star qubit.

VQNHE++ introduces a tunable gauge $W_\tau$ and tracks the transformed Hamiltonian
\begin{equation}
H_\tau = W_\tau^\dagger H W_\tau,
\end{equation}
on the operator side. Practically, this means that the coefficients used to compile the DNP estimator become $\tau$-dependent. Crucially, however, once $\tau$ is fixed, the subsequent estimation step is exactly the same as standard DNP applied to the same measurement data: the measurement pipeline and the diagonal reweighting by $f$ remain unchanged, and the energy is still computed as the same normalized ratio-of-expectations, but now with the Pauli expansion and associated tables taken from $H_\tau$ rather than $H$.

This perspective makes the algorithmic structure explicit: VQNHE++ composes a classical transformation module with a DNP evaluator, so that DNP appears as a subroutine inside a larger variational loop. When optimizing for a fixed Hamiltonian, VQNHE++ reduces to VQNHE as soon as $\tau$ is fixed: the DNP code sees constant operator-side coefficients and the optimization over $f$ (and any circuit parameters) is precisely the original form. More generally, one may alternate between updating $(\boldsymbol{\theta},f)$ at fixed $\tau$ (pure DNP on a fixed operator representation) and updating $\tau$ at fixed $(\boldsymbol{\theta},f)$, or co-optimize them if desired.

To absorb both ideas into an existing DNP implementation: (i) given $\tau$, regenerate the operator-side coefficients from $H_\tau$; then (ii) call the same DNP estimator/evaluator with those tables and the same diagonal weights. In all cases, the core non-unitary post-processing and the normalization-dependent ratio structure remain unchanged; VQNHE++ augments DNP by adding a classical, variational transformation of the Hamiltonian representation.

\section{Details for the DNP dynamic-range obstruction}
\label{app:dnp_range_obstruction}

This appendix provides the technical details supporting Sec.~\ref{subsec:dnp_range_obstruction}. The purpose is to justify the exact DNP range identity, the singular-support regularization, the Haar noncoverage bound, the Bhattacharyya certificate, and the constant-depth near-tensorizability condition.

\paragraph{Exact range identity and singular support}

The diagonal non-unitary filter acts on the ansatz state as
\begin{equation}
|\psi_f\rangle
=
\frac{1}{\sqrt{Z_f}}
\sum_{s\in\Omega_n}
f(s)\langle s|\psi\rangle |s\rangle ,
\label{eq:dnp_filtered_state_app}
\end{equation}
where \(Z_f=\sum_{s\in\Omega_n}|f(s)|^2q_n(s)\). Thus the induced Born distribution is
\begin{equation}
p_f(s)
=
\frac{|f(s)|^2 q_n(s)}
{\sum_{t\in\Omega_n}|f(t)|^2 q_n(t)} .
\label{eq:dnp_induced_born_app}
\end{equation}
Writing \(w(s)=|f(s)|^2\), exact DNP matching requires
\begin{equation}
p_n(s)
=
\frac{w(s)q_n(s)}
{\sum_{t\in\Omega_n}w(t)q_n(t)} .
\label{eq:dnp_exact_matching_app}
\end{equation}
The exact DNP range is therefore written as
\[
\Gamma^\star(p_n\mid q_n)
=
\inf_w
\left\{
\begin{aligned}
&\frac{\max_s w(s)}{\min_s w(s)}
\;:\; \\
&p_n(s)
=
\frac{w(s)q_n(s)}
{\sum_t w(t)q_n(t)}
\end{aligned}
\right\},
\label{eq:exact_dnp_range_app}
\]
as given in Eq.~\eqref{eq:dnp_exact_range}.
If \(|f(s)|\in[1/r,r]\), then \(w(s)\in[1/r^2,r^2]\), and hence
\[
\frac{\max_s w(s)}{\min_s w(s)}
\le
r^4 .
\]
Therefore exact DNP matching requires
\begin{equation}
r
\ge
\Gamma^\star(p_n\mid q_n)^{1/4}.
\label{eq:dnp_amplitude_range_requirement_app}
\end{equation}

For full-support distributions, the matching condition implies
\[
\frac{p_n(s)}{q_n(s)}
=
\frac{w(s)}{Z_w},
\qquad
Z_w=\sum_t w(t)q_n(t).
\]
Thus \(w(s)=Z_w p_n(s)/q_n(s)\). Since \(Z_w\) is independent of \(s\), it
cancels in the range, giving
\begin{equation}
\Gamma^\star(p_n\mid q_n)
=
\frac{\max_s p_n(s)/q_n(s)}
{\min_s p_n(s)/q_n(s)} .
\label{eq:exact_range_identity_app}
\end{equation}

If \(q_n\) has zero-probability components, exact matching is singular. If \(q_n(s)=0\) and \(p_n(s)>0\), then no finite positive reweighting can realize the transformation, so \(\Gamma^\star(p_n\mid q_n)=\infty\). To extend the theory continuously to this case, define
\[
Z_n=\{s:q_n(s)=0\},
\qquad
z_n=|Z_n|.
\]
For \(\varepsilon>0\), let
\[
\widetilde q_{n,\varepsilon}(s)
=
q_n(s)+\varepsilon \mathbf{1}_{s\in Z_n},
\qquad
q_{n,\varepsilon}(s)
=
\frac{\widetilde q_{n,\varepsilon}(s)}
{1+\varepsilon z_n}.
\]
Then \(q_{n,\varepsilon}\) has full support and satisfies
\begin{equation}
\|q_{n,\varepsilon}-q_n\|_{\mathrm{TV}}
=
\frac{\varepsilon z_n}{1+\varepsilon z_n},
\label{eq:singular_regularization_tv_app}
\end{equation}
where \(\|p-q\|_{\mathrm{TV}}=\frac{1}{2}\sum_s |p(s)-q(s)|\) represents the total variation distance. Thus, if \(\varepsilon z_n=o(1)\), the regularized distribution remains infinitesimally close to \(q_n\). This regularization extends the scope of the range identity, but it does not remove the obstruction: if the target assigns non-negligible probability to configurations where the ansatz has zero or near-zero probability, the required likelihood ratio becomes very large.

\paragraph{Proof of theorem~\ref{thm:haar_range_noncoverage}}

\begin{restatement}[Theorem~\ref{thm:haar_range_noncoverage}]
Fix any full-support ansatz Born distribution \(q_n\) on
\[
\Omega_n=\{0,1\}^n,
\qquad
d=2^n .
\]
Let \(|\phi\rangle\) be Haar-random in \(\mathbb C^d\), and let
\[
p_\phi(s)=|\langle s|\phi\rangle|^2 .
\]
Then, for every polynomial range
\[
R_n=\operatorname{poly}(n),
\]
we have
\[
\Pr_\phi
\left[
\Gamma^\star(p_\phi\mid q_n)\le R_n
\right]
\le
2^{-n}
+
\exp[-2^{n/2}/2]
\]
for all sufficiently large \(n\).
\end{restatement}

\begin{proof}
If \(q_n\) has a zero-probability component, then a Haar-random target assigns nonzero probability to that component almost surely. Hence no finite positive DNP reweighting can realize \(p_\phi\) from \(q_n\), and the finite-range reachable probability is zero. We therefore prove the stated bound for the full-support case.

Let \(d=2^n\). A Haar-random computational-basis probability vector can be
represented as
\begin{equation}
p_\phi(s)
=
\frac{Y_s}{T},
\qquad
T
=
\sum_{t=1}^d Y_t,
\label{eq:haar_exp_representation_app}
\end{equation}
where \(Y_1,\ldots,Y_d\) are mutually independent exponential random variables
with rate \(1\), i.e.,
\[
\Pr[Y_i>y]=e^{-y}
\qquad
(y\ge 0).
\]
Since \(q_n\) has full support, the exact range identity gives
\begin{equation}
\Gamma^\star(p_\phi\mid q_n)
=
\frac{\max_s Y_s/q_n(s)}
{\min_s Y_s/q_n(s)} .
\label{eq:haar_exact_range_exp_app}
\end{equation}

Define \(S=\{s:q_n(s)\le 2/d\}\). Then \(|S|\ge d/2\). Indeed, if
\(|S|<d/2\), then more than \(d/2\) entries satisfy \(q_n(s)>2/d\), which would
imply \(\sum_s q_n(s)>1\), a contradiction. For every \(s\in S\), we have
\(1/q_n(s)\ge d/2\), and therefore
\begin{equation}
\max_s \frac{Y_s}{q_n(s)}
\ge
\frac{d}{2}\max_{s\in S}Y_s .
\label{eq:haar_numerator_reduction_app}
\end{equation}

We first lower bound the numerator. Since
\[
\Pr
\left[
Y_s<\frac{1}{2}\log d
\right]
=
1-d^{-1/2},
\]
independence gives
\[
\Pr
\left[
\max_{s\in S}Y_s<\frac{1}{2}\log d
\right]
=
(1-d^{-1/2})^{|S|}
\le
\exp(-\sqrt d/2).
\]
Therefore, with probability at least \(1-\exp(-\sqrt d/2)\),
\begin{equation}
\max_s \frac{Y_s}{q_n(s)}
\ge
\frac{d}{4}\log d .
\label{eq:haar_numerator_bound_app}
\end{equation}

For the denominator, consider the minimum over all computational-basis strings:
\begin{align}
\Pr\left[
\min_s \frac{Y_s}{q_n(s)} > t
\right]
&=
\Pr\left[
Y_s > t q_n(s)\ \text{for all }s
\right] \nonumber \\
&=
\prod_s \Pr\left[
Y_s > t q_n(s)
\right] \nonumber \\
&=
\prod_s e^{-tq_n(s)}=e^{-t}.\nonumber 
\end{align}
Choosing \(t=\log d\), we obtain
\begin{equation}
\Pr
\left[
\min_s \frac{Y_s}{q_n(s)}\le \log d
\right]
=
1-\frac{1}{d}.
\label{eq:haar_denominator_bound_app}
\end{equation}

By a union bound, with probability at least
\(1-d^{-1}-\exp(-\sqrt d/2)\), both
Eqs.~\eqref{eq:haar_numerator_bound_app} and
\eqref{eq:haar_denominator_bound_app} hold. On this event,
\[
\Gamma^\star(p_\phi\mid q_n)
\ge
\frac{(d/4)\log d}{\log d}
=
\frac{d}{4}.
\]
Hence, for every \(R_n<d/4\),
\begin{equation}
\Pr_\phi
\left[
\Gamma^\star(p_\phi\mid q_n)\le R_n
\right]
\le
\frac{1}{d}
+
\exp(-\sqrt d/2).
\label{eq:haar_range_general_bound_app}
\end{equation}
Finally, since \(d=2^n\), any polynomial range satisfies
\(R_n=\operatorname{poly}(n)<d/4\) for all sufficiently large \(n\). Therefore,
\begin{equation}
\Pr_\phi
\left[
\Gamma^\star(p_\phi\mid q_n)\le R_n
\right]
\le
2^{-n}
+
\exp[-2^{n/2}/2].
\label{eq:haar_poly_range_bound_app}
\end{equation}
\end{proof}

\paragraph{Details for theorem~\ref{thm:near_tensorizable_bc}}
\label{app:exp_decay_proof}

We first restate Theorem~\ref{thm:near_tensorizable_bc} in a concrete form. Consider a one-dimensional geometrically local ansatz circuit of constant depth \(d\), applied to a product input state. Let \(\ell=O(d)\) be its backward light-cone width. Choose kept blocks
\[
B_1,\ldots,B_m
\]
of fixed size \(b=O(1)\), separated by buffer regions \(G\) of size \(g\). We assume \(g>2\ell\) and \(m=\Theta(n)\). Let \(K=\bigcup_{i=1}^m B_i\) denote the union of the kept blocks, and let \(R\) be the classical erasure map that discards all buffer coordinates. That is, for a probability distribution \(p_n\) on \(\Omega_n=\{0,1\}^n\),
\[
(Rp_n)(x_1,\ldots,x_m)
=
\sum_{y:\,(x_1,\ldots,x_m,y)\in\Omega_n}
p_n(x_1,\ldots,x_m,y),
\]
where \(x_i\in\{0,1\}^{B_i}\) is the bit string on the kept block \(B_i\), and \(y\) denotes all discarded buffer bits. Define
\begin{equation}
\bar p_n=Rp_n,
\qquad
\bar q_n=Rq_n.
\end{equation}
Let \(p_i\) and \(q_i\) be the \(i\)-th kept-block marginals of \(\bar p_n\) and \(\bar q_n\), respectively. We now restate the Theorem~\ref{thm:near_tensorizable_bc} in a concrete form.

\begin{restatement}[Theorem~\ref{thm:near_tensorizable_bc}]
Under the block construction above, the buffer-erased ansatz distribution
factorizes as
\[
\bar q_n(x_1,\ldots,x_m)
=
\prod_{i=1}^m q_i(x_i).
\]
Assume that the buffer-erased target distribution is near-tensorizable relative
to the accumulated local mismatch: there exists a constant \(\chi\ge 0\) such
that
\begin{equation}
\bar p_n(x_1,\ldots,x_m)
\le
e^{\chi m}
\prod_{i=1}^m p_i(x_i),
\label{eq:tensorization_decomposition}
\end{equation}
and
\begin{equation}
\frac{1}{m}
\sum_{i=1}^m
D_{1/2}(p_i\Vert q_i)
-
\chi
\ge
c>0.
\label{eq:near-tensorizability}
\end{equation}
Then
\begin{equation}
B(p_n,q_n)
\le
e^{-cm/2}
=
e^{-\Omega(n)}.
\end{equation}
Consequently,
\begin{equation}
\Gamma^\star(p_n\mid q_n)
\ge
e^{2cm}
=
e^{\Omega(n)}.
\end{equation}
\end{restatement}

\begin{proof}
Because the circuit has constant depth and geometrically local gates, every kept block \(B_i\) has a backward light cone of radius at most \(\ell=O(d)\). Since the buffers satisfy \(g>2\ell\), the backward light cones of distinct kept blocks are disjoint. The input state is a product state, so diagonal measurement probabilities on disjoint light cones factorize. Therefore, after
the buffers are erased,
\[
\bar q_n(x_1,\ldots,x_m)
=
\prod_{i=1}^m q_i(x_i).
\]

By definition of Bhattacharyya coefficient we have,
\[
B
\left(
\bar p_n,
\bigotimes_{i=1}^m q_i
\right)
=
\sum_{x_1,\ldots,x_m}
\sqrt{
\bar p_n(x_1,\ldots,x_m)
\prod_{i=1}^m q_i(x_i)
}.
\]
Using Eq.~\eqref{eq:tensorization_decomposition}, we obtain
\[
B
\left(
\bar p_n,
\bigotimes_{i=1}^m q_i
\right)
\le
e^{\chi m/2}
\sum_{x_1,\ldots,x_m}
\prod_{i=1}^m
\sqrt{p_i(x_i)q_i(x_i)}.
\]
The sum over \(x_1,...,x_m\) factorizes over the remaining blocks, giving
\[
B
\left(
\bar p_n,
\bigotimes_{i=1}^m q_i
\right)
\le
e^{\chi m/2}
\prod_{i=1}^m B(p_i,q_i).
\]
Taking \(-2\log\) of both sides yields
\[
D_{1/2}
\left(
\bar p_n
\middle\Vert
\bigotimes_{i=1}^m q_i
\right)
\ge
\sum_{i=1}^m D_{1/2}(p_i\Vert q_i)
-
\chi m.
\]
Since
\[
\bar q_n
=
\bigotimes_{i=1}^m q_i,
\]
the near-tensorizability condition~\eqref{eq:near-tensorizability} gives
\[
D_{1/2}(\bar p_n\Vert \bar q_n)
\ge
cm.
\]

Finally, the erasure map \(R\) is a classical stochastic map, so the
R\'enyi-\(1/2\) divergence decreases under the erasure map:
\[
D_{1/2}(p_n\Vert q_n)
\ge
D_{1/2}(\bar p_n\Vert \bar q_n)
\ge
cm.
\]
Using
\[
D_{1/2}(p_n\Vert q_n)
=
-2\log B(p_n,q_n),
\]
we obtain
\[
B(p_n,q_n)
\le
e^{-cm/2}.
\]
The Bhattacharyya certificate given by Prop.~\ref{prop:bc_range_certificate},
\[
\Gamma^\star(p_n\mid q_n)
\ge
B(p_n,q_n)^{-4},
\]
then gives
\[
\Gamma^\star(p_n\mid q_n)
\ge
e^{2cm}.
\]
Since \(m=\Theta(n)\), this is exponential in the system size.
\end{proof}

The near-tensorizability condition~\eqref{eq:near-tensorizability} compares the residual dependence between the retained blocks with the local mismatch between the target and the ansatz. It does not require the target distribution to factorize exactly. Instead, the pointwise correlation contribution, quantified by $\chi$, must remain smaller than the average block-wise divergence by a positive margin as the system size increases. Correlations within each block are unrestricted. The condition depends on the block partition and the chosen ansatz, and serves as a sufficient criterion identifying target--ansatz pairs to which the theorem applies.

A transverse-field Ising chain with alternating couplings illustrates how this correlation contribution depends on the Hamiltonian. For a numerical example, we consider the 14-qubit model
\begin{equation}
H=
-\sum_{j=0}^{6}Z_{2j}Z_{2j+1}
-\lambda\sum_{j=0}^{5}Z_{2j+1}Z_{2j+2}
-h\sum_{j=0}^{13}X_j,
\end{equation}
with open boundaries, unit coupling within each pair $(2j,2j+1)$, and coupling $\lambda$ between neighboring pairs. From the exact ground-state Born distribution, we retain three blocks, $B_1=(2,3)$, $B_2=(6,7)$, and $B_3=(10,11)$, and sum over the measurement outcomes of the remaining qubits.

The parameters $\lambda$ and $h$ control the target correlations in complementary ways. We plot the minimum value of $\chi$ required for the pointwise factorization bound to hold, $\chi_{\min}=\frac{1}{m}\max_{\bar p(x)>0}\log[\bar p(x)/\prod_i p_i(x_i)]$, with $m=3$ retained blocks in Figure~\ref{fig:chi_locality}. Figure~\ref{fig:chi_locality} shows $\chi_{\min}$ across the $(\lambda,h)$ parameter plane. The correlation contribution is small when the coupling between neighboring pairs is weak or the transverse field is sufficiently strong. At $\lambda=0$ and $h>0$, the ground state factorizes into independent two-qubit states, giving $\chi_{\min}=0$ despite correlations within each pair. Increasing $h$ also suppresses the correlation contribution as the ground state approaches $|+\rangle^{\otimes14}$. These trends illustrate how the Hamiltonian can favor approximate factorization between the retained blocks. Note that these finite-size results illustrate the target-dependent correlation contribution; they do not establish the full condition, which additionally requires an ansatz-dependent local mismatch exceeding this contribution by a positive margin as the system size grows.

\begin{figure}[t]
    \centering
    \includegraphics[width=0.9\linewidth]{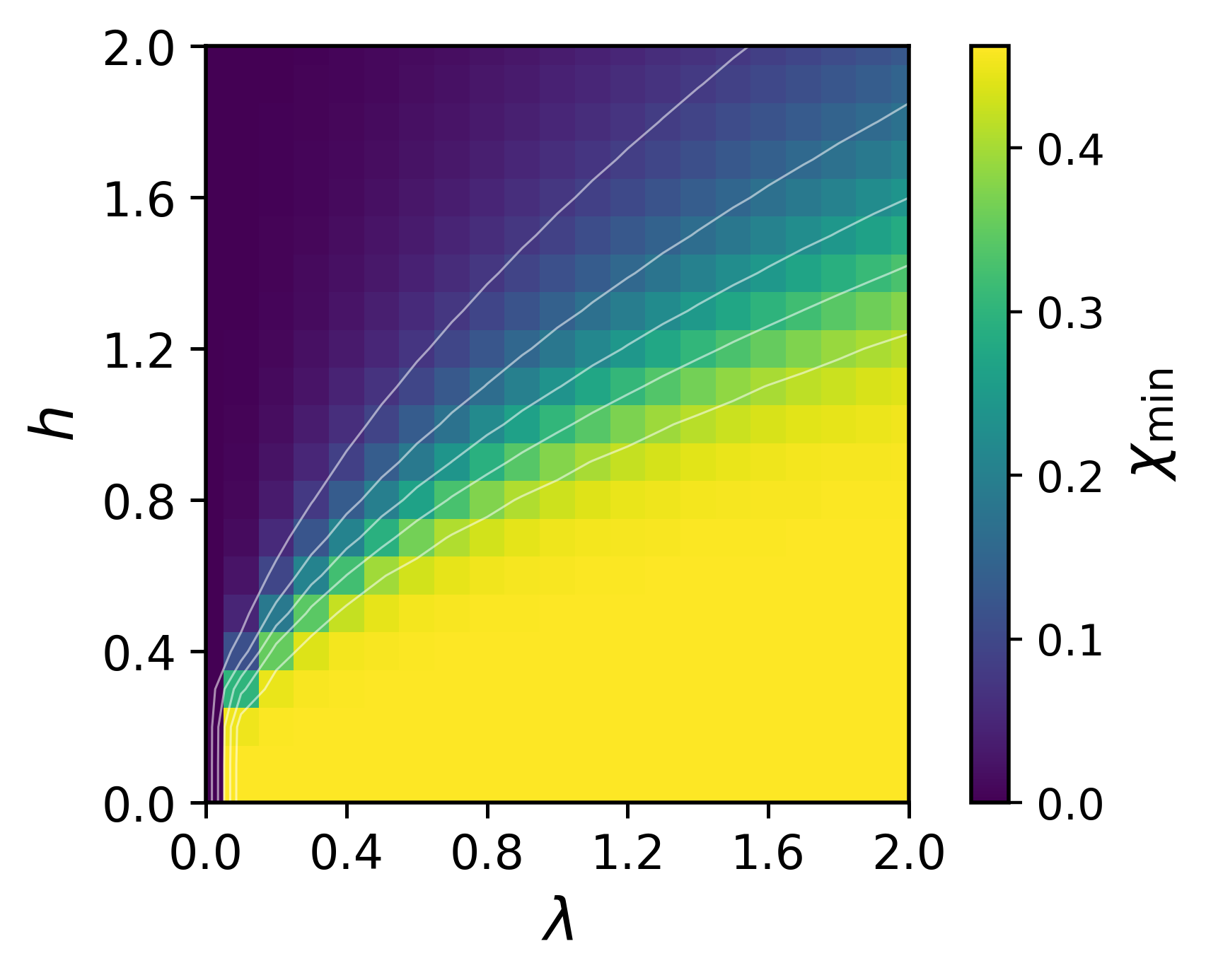}
    \caption{\justifying
    Pointwise correlation contribution $\chi_{\min}$ for the 14-qubit transverse-field Ising chain with alternating couplings, as a function of the coupling $\lambda$ between neighboring pairs and the transverse field $h$. Ground-state probabilities are obtained by exact diagonalization. Colors indicate $\chi_{\min}$, and solid lines mark contours of constant $\chi_{\min}$. Weakening the coupling between pairs or increasing the transverse field reduces the correlation contribution.}
    \label{fig:chi_locality}
\end{figure}

A similar physical motivation arises in quantum chemistry when electronic correlations are strong within localized orbital groups (or active spaces) but weaker between them. Classical density-matrix renormalization group calculations exploit this structure in quasi-one-dimensional orbital structures. For example, correlations within a bond in a hydrogen chain can be strong while those between distant bonds are weaker~\cite{ChanSharma2011DMRG,Hachmann2006DMRG}. VQE-based approaches also exploit this separation by first treating the correlations within individual subsystems and then accounting for their coupling. Deep VQE applies this strategy to systems with strong interactions within subsystems and weaker interactions between them~\cite{Fujii2022}. In quantum chemistry, localized active-space VQE methods combine correlated fragment states and optimize parametrized excitation operators involving different fragments to recover interfragment correlations~\cite{Otten2022,Mitra2024}. This separate treatment of internal and interfragment correlations motivates the structure of our assumption.

Trainability provides an additional motivation for focusing on this regime. Capturing correlations over long distances generally requires deeper local circuits, and sufficiently random, expressive circuit ensembles can exhibit barren plateaus~\cite{McClean2018,Holmes2022}. Shallow circuits with local cost functions can avoid exponentially suppressed gradients under suitable conditions~\cite{Cerezo2021}, making targets with predominantly local correlations relevant to efforts to balance accuracy and trainability. This connection motivates the choice of problems; near-tensorizability alone does not guarantee that a given ansatz can be trained efficiently.

\section{Finite-shot stability of U-VQNHE}
\label{app:uvqnhe_finite_shot}

This appendix gives the finite-shot stability analysis for U-VQNHE. The key difference from diagonal non-unitary post-processing is that U-VQNHE uses a diagonal unitary phase map rather than a non-unitary amplitude reweighting. As a result, the post-processed state remains normalized by construction, and the finite-shot estimator does not contain a stochastic normalization denominator.

The U-VQNHE phase map is
\begin{equation}
U_g
=
\sum_{s\in\Omega_n}
e^{i g(s)}
|s\rangle\langle s|,
\qquad
g(s)\in\mathbb R .
\label{eq:uvqnhe_phase_map_app}
\end{equation}
For a normalized ansatz state \(|\psi\rangle\), the post-processed state is
\[
|\psi_g\rangle
=
U_g|\psi\rangle .
\]
Since \(|e^{ig(s)}|=1\), the map \(U_g\) is unitary. Therefore
\[
\| |\psi_g\rangle \|
=
\| |\psi\rangle \|
=
1 .
\]
Consequently, for any Hamiltonian \(H\) with ground-state energy \(E_{\rm gs}\), the exact U-VQNHE energy
\begin{equation}
E_g
=
\langle\psi|U_g^\dagger H U_g|\psi\rangle
\label{eq:uvqnhe_exact_energy_app}
\end{equation}
satisfies the variational bound
\begin{equation}
E_g
\ge
E_{\rm gs}.
\label{eq:uvqnhe_exact_variational_bound_app}
\end{equation}
This is the exact variational safety of U-VQNHE.

We now describe the finite-shot estimator. Let the Hamiltonian be decomposed into Pauli terms as in Eq.~\eqref{eq:pauli_hamiltonian_app}. For each Pauli term \(P\), the corresponding U-VQNHE expectation value can be estimated from one or two real-valued circuit components, each for real and imaginary values of the unitary transformation, indexed by \(\tau\in T_P\subseteq\{R,I\}\). For each component, write
\begin{equation}
\mu_{P,\tau}
=
\mathbb E_{s\sim \pi_{P,\tau}}
\left[
X_{P,\tau}(s)
\right],
\qquad
|X_{P,\tau}(s)|\le 1 .
\label{eq:uvqnhe_component_expectation_app}
\end{equation}
The boundedness follows from the fact that the learned factor is a phase
\(e^{ig(s)}\), whose modulus is one. In particular, the real and imaginary components are bounded by \(1\) in absolute value.

The exact energy can then be written as
\begin{equation}
E_g
=
\sum_{P\in\mathcal P}
c_P
\sum_{\tau\in T_P}
\mu_{P,\tau}.
\label{eq:uvqnhe_energy_component_sum_app}
\end{equation}
Given \(N\) independent shots for component \((P,\tau)\), define
\[
\widehat{\mu}_{P,\tau}
=
\frac{1}{N}
\sum_{k=1}^N
X_{P,\tau}(s_k),
\]
and
\begin{equation}
\widehat E_g
=
\sum_{P\in\mathcal P}
c_P
\sum_{\tau\in T_P}
\widehat{\mu}_{P,\tau}.
\label{eq:uvqnhe_finite_shot_estimator_app}
\end{equation}

\begin{proposition}[Unbiasedness and variance of the U-VQNHE estimator]
\label{prop:uvqnhe_variance_app}
Assume that the shot ensembles are independent across all components \((P,\tau)\), and that each component is estimated with \(N\) independent shots. Then \(\widehat E_g\) is an unbiased estimator of \(E_g\), and
\begin{equation}
\operatorname{Var}(\widehat E_g)
\le
\frac{2C}{N},
\qquad
C=
\sum_{P\in\mathcal P}c_P^2 .
\label{eq:uvqnhe_equal_shot_variance_app}
\end{equation}
\end{proposition}

\begin{proof}
Unbiasedness follows from linearity of expectation. Since
\[
\mathbb E[\widehat{\mu}_{P,\tau}]
=
\mu_{P,\tau},
\]
we have
\[
\mathbb E[\widehat E_g]
=
\sum_{P\in\mathcal P}
c_P
\sum_{\tau\in T_P}
\mathbb E[\widehat{\mu}_{P,\tau}]
=
\sum_{P\in\mathcal P}
c_P
\sum_{\tau\in T_P}
\mu_{P,\tau}
=
E_g .
\]

For the variance, independence across components gives
\[
\operatorname{Var}(\widehat E_g)
=
\sum_{P\in\mathcal P}
c_P^2
\sum_{\tau\in T_P}
\operatorname{Var}(\widehat{\mu}_{P,\tau}) .
\]
Since each component is estimated with \(N\) independent shots,
\[
\operatorname{Var}(\widehat{\mu}_{P,\tau})
=
\frac{\sigma_{P,\tau}^2}{N},
\qquad
\sigma_{P,\tau}^2
=
\operatorname{Var}_{s\sim\pi_{P,\tau}}
\left[
X_{P,\tau}(s)
\right].
\]
Therefore
\[
\operatorname{Var}(\widehat E_g)
=
\frac{1}{N}
\sum_{P\in\mathcal P}
c_P^2
\sum_{\tau\in T_P}
\sigma_{P,\tau}^2 .
\]
The boundedness \(|X_{P,\tau}(s)|\le 1\) implies
\(\sigma_{P,\tau}^2\le 1\), so
\[
\operatorname{Var}(\widehat E_g)
\le
\frac{1}{N}
\sum_{P\in\mathcal P}
c_P^2
|T_P|.
\]
Since every Pauli term is associated with at most two estimators,
\[
|T_P|\le 2,
\]
and hence
\[
\operatorname{Var}(\widehat E_g)
\le
\frac{2}{N}
\sum_{P\in\mathcal P}c_P^2
=
\frac{2C}{N}.
\]
\end{proof}

The variance bound shows that U-VQNHE has ordinary sample-mean behavior, meaning that the finite-shot error decays as \(O(N^{-1/2})\) without any additional random
normalization denominator.

Next, we give a high-probability finite-shot bound from Hoeffding's inequality.

\begin{proposition}[Hoeffding concentration for U-VQNHE]
\label{prop:uvqnhe_hoeffding_app}
Under the assumptions of Proposition~\ref{prop:uvqnhe_variance_app}, with
probability at least \(1-\delta\),
\begin{equation}
|\widehat E_g-E_g|
\le
2\|c\|_2
\sqrt{\frac{\log(2/\delta)}{N}},
\label{eq:uvqnhe_lcb_equal_shot_app}
\end{equation}
where \(\|c\|_2^2 = \sum_{P\in\mathcal P}c_P^2\).
\end{proposition}

\begin{proof}
Define
\[
S_T
=
\sum_{P\in\mathcal P}
c_P^2
|T_P|.
\]
Write the finite-shot error as a sum of independent centered bounded random
variables:
\[
\widehat E_g-E_g
=
\sum_{P\in\mathcal P}
\sum_{\tau\in T_P}
\sum_{k=1}^{N}
\frac{c_P}{N}
\left(
X_{P,\tau}(s_k)-\mu_{P,\tau}
\right).
\]
Since \(X_{P,\tau}(s)\in[-1,1]\), each centered summand lies in an interval of length \(\frac{2|c_P|}{N}\). Hoeffding's inequality then gives
\[
\begin{gathered}
\Pr
\left[
|\widehat E_g-E_g|\ge \epsilon
\right] \\
\le
2\exp
\left(
-\frac{2\epsilon^2}
{
\sum_{P\in\mathcal P}
\sum_{\tau\in T_P}
N
\left(
2|c_P|/N
\right)^2
}
\right).
\end{gathered}
\]
The denominator simplifies as
\[
\sum_{P\in\mathcal P}
\sum_{\tau\in T_P}
N
\left(
2|c_P|/N
\right)^2
=
\frac{4}{N}
\sum_{P\in\mathcal P}
c_P^2
|T_P|
=
\frac{4S_T}{N}.
\]
Therefore,
\[
\Pr
\left[
|\widehat E_g-E_g|\ge \epsilon
\right]
\le
2\exp
\left(
-\frac{N\epsilon^2}{2S_T}
\right).
\]
Solving the right-hand side for \(\epsilon\) gives, with probability at least
\(1-\delta\),
\[
|\widehat E_g-E_g|
\le
\sqrt{
\frac{2S_T\log(2/\delta)}{N}
}.
\]
Since every Pauli term is associated with at most two estimators, \(|T_P|\le 2\)
for all \(P\), and therefore
\[
S_T
=
\sum_{P\in\mathcal P}
c_P^2
|T_P|
\le
2\sum_{P\in\mathcal P}c_P^2
=
2\|c\|_2^2 .
\]
Substituting this into the previous inequality gives
\[
|\widehat E_g-E_g|
\le
2\|c\|_2
\sqrt{\frac{\log(2/\delta)}{N}} .
\]
\end{proof}

The exact U-VQNHE objective remains variational because \(U_g\) is unitary, so \(E_g\ge E_{\rm gs}\). A finite-shot estimate \(\widehat E_g\), however, is an ordinary statistical estimator and can fluctuate above or below \(E_g\). The important point is that this fluctuation is tightly controlled around the exact U-VQNHE energy:
\[
|\widehat E_g-E_g|
=
O\!\left(
\|c\|_2
\sqrt{\frac{\log(1/\delta)}{N}}
\right)
\]
with probability at least \(1-\delta\). Unlike DNP, this concentration bound has no dependence on a neural-network amplitude range or on a stochastic normalization denominator.

Combining the concentration bound with \(E_g\ge E_{\rm gs}\) gives a certified finite-shot statement with explicit dependence only on the shot number \(N\) and the Hamiltonian coefficient norm \(\|c\|_2\). In particular, from a measured estimate \(\widehat E_g\), one obtains the high-probability certificate
\begin{equation}
E_{\rm gs}
\le
E_g
\le
\widehat E_g
+
2\|c\|_2
\sqrt{\frac{\log(2/\delta)}{N}} .
\end{equation}
Thus U-VQNHE provides a controlled upper-bound certificate for the ground-state energy up to an explicit finite-shot error bar. In contrast, DNP in general does not give an analogous range-independent certificate, because its finite-shot behavior depends on the neural-network amplitude range and on the stochastic normalization denominator.

\section{Relation to other non-unitary variational eigensolvers}
\label{app:other_nonunitary}

DNP is not the only approach that introduces non-unitary transformations into VQE. Other related methods employ structured Jastrow correlators or diagonal transformations in the occupation-number basis. These constructions share the exact normalized objective
\begin{equation}
    E_O
    =
    \frac{
        \langle\psi|O^\dagger H O|\psi\rangle
    }{
        \langle\psi|O^\dagger O|\psi\rangle
    },
    \label{eq:general_nonunitary_objective}
\end{equation}
where \(O\) is generally non-unitary. Whenever \(\langle\psi|O^\dagger O|\psi\rangle>0\), Eq.~\eqref{eq:general_nonunitary_objective} is the energy expectation value of the normalized state and therefore satisfies the Rayleigh--Ritz bound exactly.

The finite-shot behavior nevertheless depends on how the numerator and denominator are measured and on how \(O\) is parametrized. Consequently, the sampled-support results established for DNP do not automatically apply to every non-unitary variational algorithm. Here we identify the connections that follow directly from their estimator structures, without asserting identical sample-complexity behavior.

\subsection{Structured Jastrow transformations}

The non-unitary variational quantum eigensolver~(nu-VQE) introduces a Jastrow-inspired operator acting on the circuit state~\cite{Benfenati2021}. Its polynomial implementation uses the linearized Hermitian operator
\begin{equation}
    J_{\mathrm{lin}}
    =
    I-A,
    \qquad
    A
    =
    \sum_{i=1}^{n}\alpha_i Z_i
    +
    \sum_{i<j}\lambda_{ij}Z_iZ_j ,
    \label{eq:linearized_jastrow}
\end{equation}
with real variational parameters \(\alpha_i\) and \(\lambda_{ij}\). The corresponding energy is
\begin{equation}
    E_J
    =
    \frac{
        \langle\psi|J_{\mathrm{lin}} HJ_{\mathrm{lin}}|\psi\rangle
    }{
        \langle\psi|J_{\mathrm{lin}}^2|\psi\rangle
    }.
    \label{eq:nuvqe_energy}
\end{equation}
The transformed operators must retain the complete two-sided expansion,
\begin{align}
    J_{\mathrm{lin}}HJ_{\mathrm{lin}}
    &=
    H-AH-HA+AHA,
    \label{eq:nuvqe_numerator_expansion}
    \\
    J_{\mathrm{lin}}^2
    &=
    I-2A+A^2.
    \label{eq:nuvqe_denominator_expansion}
\end{align}
Because \(A\) contains \(O(n^2)\) Pauli terms, the products in Eqs.~\eqref{eq:nuvqe_numerator_expansion} and \eqref{eq:nuvqe_denominator_expansion} can generate up to \(O(n^4)\) Pauli products per Hamiltonian term before identical terms are combined or commuting observables are grouped. This polynomial measurement overhead is the cost of representing the non-unitary transformation through transformed observables rather than through a deeper state-preparation circuit.

Finite-shot estimates of Eq.~\eqref{eq:nuvqe_energy} form a ratio of sampled expectation values. Statistical errors can therefore be amplified when the estimated normalization is small or poorly conditioned. This observation does not, by itself, imply the sampled-support instability. Unlike unrestricted DNP, the values of the Jastrow transformation are coupled through only \(O(n^2)\) parameters, and the numerator and denominator contain correlated operator products generated by the same structured transformation. Establishing a quantitative sampling bound for nu-VQE therefore requires an analysis specific to this operator family.

The Jastrow quantum circuit~(JQC) construction similarly augments a circuit state with the correlator~\cite{Mazzola2019,Jastrow1955}
\begin{equation}
    P_J=e^{J},
    \qquad
    J=
    \sum_{k\neq l}\lambda_{kl}Z_kZ_l.
    \label{eq:jqc_operator}
\end{equation}
Its exact energy has the normalized form
\begin{equation}
    E_{\mathrm{JQC}}
    =
    \frac{
        \langle\psi|P_J^\dagger H P_J|\psi\rangle
    }{
        \langle\psi|P_J^\dagger P_J|\psi\rangle
    }.
    \label{eq:jqc_energy}
\end{equation}
JQC admits different implementation strategies. An ancilla-assisted construction reconstructs and reweights measurement probabilities, whereas a transformed-Hamiltonian construction measures the Pauli expansions of \(P_J^\dagger H P_J\) and \(P_J^\dagger P_J\). For the full exponential correlator, the latter expansions can contain exponentially many Pauli terms; truncating the exponential restores polynomial operator growth at the cost of restricting the variational family~\cite{Mazzola2019}.

\subsection{Cascaded variational quantum eigensolver}

The cascaded variational quantum eigensolver~(CVQE) separates quantum sampling from classical parameter optimization~\cite{Gunlycke2024}. A fixed guiding state
\begin{equation}
    |\psi_0\rangle=U|0\rangle
\end{equation}
is sampled using a prescribed family of measurement circuits, while the variational transformation is diagonal in the occupation-number basis:
\begin{equation}
    \widehat{\lambda}(\boldsymbol{\phi})
    =
    \sum_{\boldsymbol{n}}
    \lambda_{\boldsymbol{n}}(\boldsymbol{\phi})
    |\boldsymbol{n}\rangle\langle\boldsymbol{n}|.
    \label{eq:cvqe_diagonal_operator}
\end{equation}
The trial state and its normalized energy are
\begin{align}
    |\psi(\boldsymbol{\phi})\rangle
    &=
    e^{i\widehat{\lambda}(\boldsymbol{\phi})}
    |\psi_0\rangle,
    \label{eq:cvqe_state}
    \\
    E_{\mathrm{CVQE}}(\boldsymbol{\phi})
    &=
    \frac{
        \langle\psi_0|
        e^{-i\widehat{\lambda}^{\dagger}(\boldsymbol{\phi})}
        H
        e^{i\widehat{\lambda}(\boldsymbol{\phi})}
        |\psi_0\rangle
    }{
        \langle\psi_0|
        e^{-i\widehat{\lambda}^{\dagger}(\boldsymbol{\phi})}
        e^{i\widehat{\lambda}(\boldsymbol{\phi})}
        |\psi_0\rangle
    }.
    \label{eq:cvqe_energy}
\end{align}
When all \(\lambda_{\boldsymbol n}(\boldsymbol{\phi})\) are real, \(e^{i\widehat{\lambda}(\boldsymbol{\phi})}\) is unitary and the denominator in Eq.~\eqref{eq:cvqe_energy} is exactly one. For complex \(\lambda_{\boldsymbol n}(\boldsymbol{\phi})\), however, the transformation modifies both the phase and magnitude of each occupation-number amplitude. Thus, the real and imaginary components control phase and amplitude corrections, respectively. Writing \(p_0(\boldsymbol n)=|\langle\boldsymbol n|\psi_0\rangle|^2\), the normalization of the transformed state becomes
\begin{equation}
    \Lambda(\boldsymbol{\phi})
    =
    \sum_{\boldsymbol n}
    e^{-2\,\operatorname{Im}
    \lambda_{\boldsymbol n}(\boldsymbol{\phi})}
    p_0(\boldsymbol n).
    \label{eq:cvqe_normalization}
\end{equation}
The imaginary components therefore act as configuration-dependent amplitude weights on the guiding-state distribution, while the real components supply phase corrections.

The non-unitary CVQE objective is therefore a ratio estimator with a parameter-dependent normalization constructed from probability-mass functions measured on the guiding-state circuits. Large variations in the imaginary components of \(\lambda_{\boldsymbol n}\) can assign substantial effective weight to configurations that are rare under \(p_0\), thereby reducing the effective sample size and increasing finite-shot sensitivity to the normalization. This behavior is algebraically analogous to amplitude reweighting in DNP. However, CVQE employs a distinct measurement construction and a structured parametrization. Under the locality and scaling assumptions adopted in Ref.~\cite{Gunlycke2024}, the number of measurement settings and the number of terms entering the classical energy evaluation grow at most polynomially with system size, and the collected measurement samples can be reused throughout the optimization. This establishes polynomial algebraic post-processing conditional on the available samples, but it does not provide a uniform polynomial bound on the number of shots required to estimate a potentially ill-conditioned normalization. Consequently, the DNP results motivate examining finite-shot conditioning in CVQE but do not, by themselves, establish an exponential measurement requirement for that method.

\section{Benchmark hamiltonians and experimental details}
\label{app:benchmark-details}

This appendix summarizes the benchmark Hamiltonians and experimental settings used to compare VQE, VQNHE, and U-VQNHE, except for the standard 1D TFIM. For each benchmark, a VQE reference state was first obtained by optimizing a parameterized quantum circuit. The optimized circuit parameters were then fixed, and the same reference state was used as the input state for both VQNHE and U-VQNHE. All circuit expectation values were estimated using \(N_{\mathrm{shots}}=1024\) shots per measurement circuit.

\subsection{Random-coefficient Ising model}

The first benchmark is a one-dimensional random-coefficient Ising model on
\(n=12\) qubits with open boundary conditions:
\begin{equation}
H_{\mathrm{rand}}
=
\sum_{i=1}^{n-1} J_i Z_i Z_{i+1}
+
\sum_{i=1}^{n} h_i^{x} X_i
+
\sum_{i=1}^{n} h_i^{z} Z_i .
\label{eq:random-ising}
\end{equation}
The coefficients were sampled independently as
\begin{align}
J_i &\sim \mathcal{N}(1.0,0.3^2), \\
h_i^{x} &\sim \mathcal{N}(-1.0,0.3^2), \\
h_i^{z} &\sim \mathcal{N}(0.0,0.3^2).
\end{align}
The sampled values are shown in Figure~\ref{fig:random-ising-couplings}. The resulting Hamiltonian contains \(35\) Pauli terms.

\begin{figure}[!tbp]
    \centering
    \includegraphics[width=0.9\columnwidth]{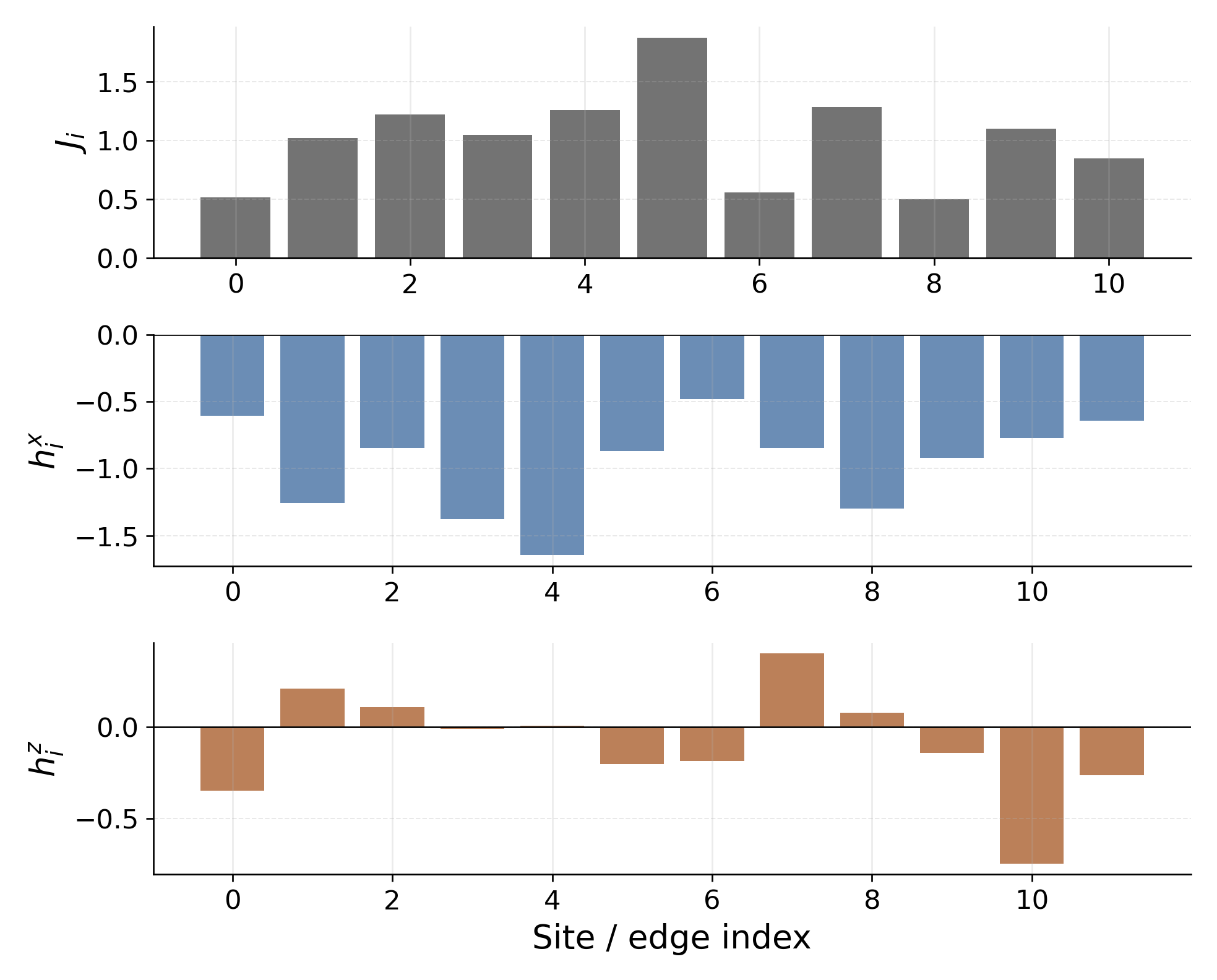}
    \caption{\justifying
    Coupling coefficients for the fixed 12-qubit random Ising realization. The upper panel shows the 11 nearest-neighbor interaction strengths $J_i$, while the middle and lower panels show the transverse fields $h_i^x$ and longitudinal fields $h_i^z$, respectively, at each of the 12 lattice sites. The same disorder realization was used for all five VQE, VQNHE, and U-VQNHE repetitions.
    }
    \label{fig:random-ising-couplings}
\end{figure}

The VQE reference circuit is a one-layer Ising hardware-efficient ansatz. The circuit begins with a Hadamard gate on every qubit and then applies
\begin{equation}
\begin{aligned}
&U_{\mathrm{Ising}}(\boldsymbol{\theta}) \\
&=
\left[
\prod_{i=1}^{n} R_x^{(i)}(\beta_i)
\right]
\left[
\prod_{i=1}^{n-1}
\mathrm{CX}_{i,i+1}
R_z^{(i+1)}(\alpha_i)
\mathrm{CX}_{i,i+1}
\right] .
\end{aligned}
\label{eq:ising-ansatz}
\end{equation}
For \(n=12\), this circuit contains \(2n-1=23\) variational parameters.

\subsection{Two-dimensional transverse-field Ising model}

The second benchmark is a transverse-field Ising model on a \(3\times4\) rectangular lattice:
\begin{equation}
H_{\mathrm{2D}}
=
J\sum_{\langle i,j\rangle} Z_iZ_j
+
h_x\sum_{i=1}^{12}X_i,
\label{eq:2d-tfim}
\end{equation}
where \(\langle i,j\rangle\) denotes horizontal and vertical nearest-neighbor pairs. We use open boundary conditions in both lattice directions, with \(J=1\) and \(h_x=-1\). The lattice contains \(17\) interaction edges, so the Hamiltonian has \(29\) Pauli terms in total.

The same one-layer Ising ansatz in Eq.~\eqref{eq:ising-ansatz} was used, again with \(23\) circuit parameters. The entangling gates follow a one-dimensional nearest-neighbor chain, while the Hamiltonian has two-dimensional connectivity. This choice keeps the circuit depth and parameter count identical to those of the random Ising benchmark while making the target Hamiltonian structurally more demanding.

\subsection{Random two-dimensional XYZ model}
\label{app:random-xyz}

The third benchmark is a disordered XYZ spin model on a \(3\times3\) square lattice with open boundary conditions. The Hamiltonian is
\begin{equation}
H_{\mathrm{XYZ}}
=
\sum_{\langle i,j\rangle}
\sum_{\alpha\in\{x,y,z\}}
J_{ij}^{\alpha}
\sigma_i^\alpha\sigma_j^\alpha
+
\sum_{i=1}^{9}
\sum_{\alpha\in\{x,y,z\}}
h_i^\alpha\sigma_i^\alpha ,
\label{eq:random-xyz}
\end{equation}
where \(\langle i,j\rangle\) denotes horizontal and vertical nearest-neighbor pairs, and
\(\sigma_i^x=X_i,\sigma_i^y=Y_i,\sigma_i^z=Z_i\). The \(3\times3\) open lattice contains \(12\) nearest-neighbor edges. Each edge contributes three two-qubit interaction terms, and each of the nine sites contributes three single-qubit field terms. The Hamiltonian therefore contains \(3\times12+3\times9=63\) Pauli terms.

The anisotropic exchange coefficients were sampled independently as \(J_{ij}^{\alpha}\sim\mathcal{N}(0,1)\), and the local fields were sampled independently as \(h_i^\alpha\sim\mathcal{N}(0,0.5^2)\), whose values are shown in Figure~\ref{fig:random-xyz-couplings}. The zero-mean exchange distribution produces both positive and negative couplings, introducing frustration. Independent coefficients were used for the \(XX\), \(YY\), and \(ZZ\) interactions, so the Hamiltonian is not globally spin-rotation invariant.

\begin{figure}[tb]
    \centering
    \includegraphics[width=0.97\columnwidth]{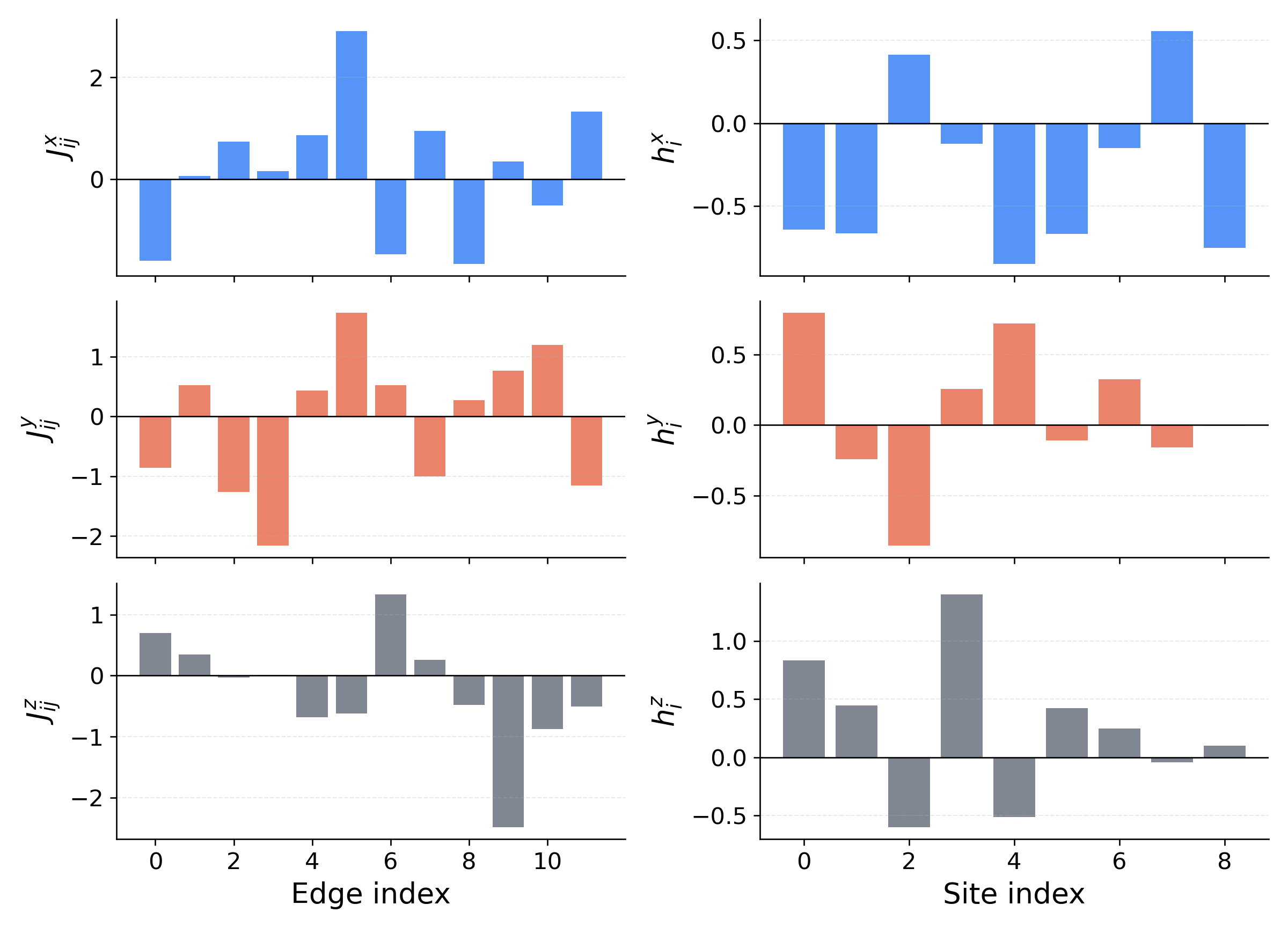}
    \caption{\justifying
    Exchange couplings and local fields for the random XYZ realization on the $3\times3$ open lattice. The left column shows the independently sampled nearest-neighbor exchange coefficients $J_{ij}^{x}$, $J_{ij}^{y}$, and $J_{ij}^{z}$ for the 12 lattice edges. The right column shows the corresponding local fields $h_i^{x}$, $h_i^{y}$, and $h_i^{z}$ at the nine lattice sites. The realization was held fixed across all five VQE, VQNHE, and U-VQNHE repetitions.
    }
    \label{fig:random-xyz-couplings}
\end{figure}

A single disorder realization was fixed across all optimization trials. Thus, the reported standard deviations measure finite-shot and optimization variability for one Hamiltonian instance, rather than variation over disorder realizations.

The independent random couplings remove translation and lattice-reflection symmetries. The simultaneous presence of random fields along all three spin directions also breaks global spin-flip symmetry, total magnetization conservation, and continuous spin-rotation symmetry. In particular, the single-qubit \(Y\)-field terms make the Hamiltonian generally complex in the computational basis. This benchmark therefore provides a useful test of the phase transformation used in U-VQNHE.

For this benchmark, we used a symmetry-unconstrained two-dimensional hardware-efficient ansatz. The circuit starts from \(\ket{0}^{\otimes 9}\) and applies independent \(R_y\) and \(R_z\) rotations to every qubit:
\begin{equation}
U_{\mathrm{rot}}^{(\ell)}
=
\prod_{i=1}^{9}
R_z^{(i)}(\phi_{\ell i})
R_y^{(i)}(\theta_{\ell i}) .
\label{eq:xyz-rotation-layer}
\end{equation}
Each entangling layer applies controlled-\(Z\) gates on every horizontal and vertical nearest-neighbor edge as
\(U_{\mathrm{ent}}=\prod_{\langle i,j\rangle}\mathrm{CZ}_{i,j}\). For entangling depth \(L=2\), the full circuit is
\begin{equation}
U_{\mathrm{XYZ}}
=
U_{\mathrm{rot}}^{(2)}
U_{\mathrm{ent}}
U_{\mathrm{rot}}^{(1)}
U_{\mathrm{ent}}
U_{\mathrm{rot}}^{(0)} .
\label{eq:xyz-ansatz}
\end{equation}
There are \(L+1=3\) rotation layers, including the initial layer, with two parameters per qubit in each layer. The total number of circuit parameters is \(N_{\mathrm{circ}}=2\times9\times(L+1)=54\). Because every qubit receives independent \(R_y\) and \(R_z\) rotations, the ansatz does not enforce magnetization, spin-flip, particle-number, or lattice symmetry. The two-dimensional entangling pattern follows the connectivity of the Hamiltonian.

\subsection{Amplitude and phase contributions to the VQE energy error}
\label{app:phase_amplitude_analysis}

We analyze the trained VQE states for TFIM instances with $n=7,\ldots,14$ qubits to quantify the energy improvement available through phase-only post-processing. For each circuit state, define $E_+(q)$ as the minimum energy over all choices of phases in the exact VQE statevector, with its amplitude magnitudes and Born distribution $q$ held fixed. For the TFIM convention with negative transverse-field coefficients, all off-diagonal Hamiltonian matrix elements are nonpositive. Equal phases therefore minimize all off-diagonal energy contributions simultaneously. We obtain $E_+(q)$ directly by replacing the statevector coefficients with their absolute values and evaluating the energy, without additional numerical optimization.

The VQE energy error then separates into the phase-- and amplitude--dependent parts:
\begin{align}
E_{\mathrm{VQE}}-E_{\mathrm{gs}}
&=\Delta E_{\mathrm{phase}}+\Delta E_{\mathrm{res}},\\
\Delta E_{\mathrm{phase}}
&=E_{\mathrm{VQE}}-E_+(q),\\
\Delta E_{\mathrm{res}}
&=E_+(q)-E_{\mathrm{gs}}.
\end{align}
Here, $\Delta E_{\mathrm{phase}}$ is the error removable by optimal phase corrections, while $\Delta E_{\mathrm{res}}$ is the residual error imposed by the fixed Born distribution. The latter cannot be removed by phase learning alone, even with an unrestricted phase function.

\begin{figure}[t]
    \centering
    \includegraphics[width=\linewidth]{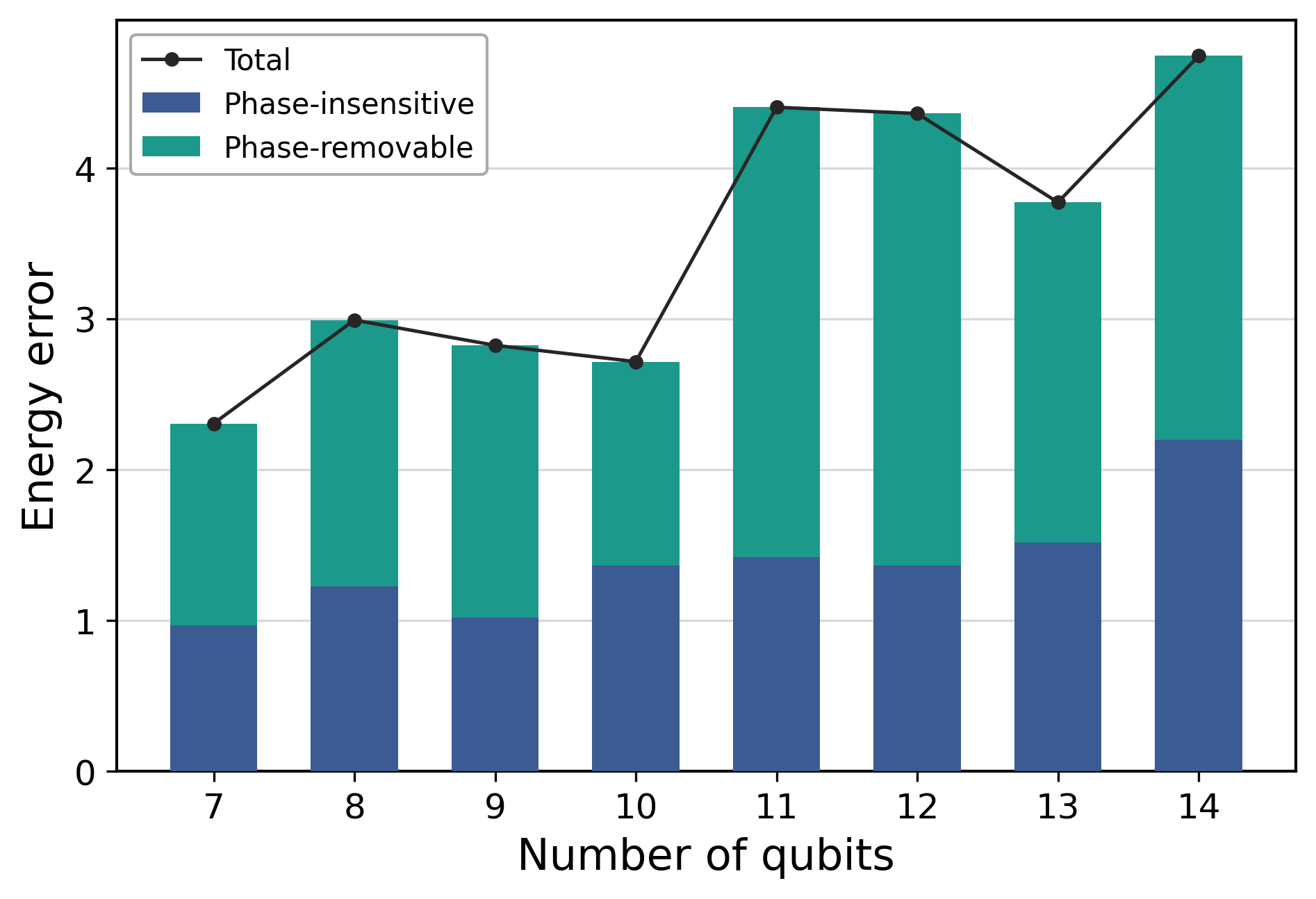}
    \caption{\justifying
    Decomposition of the VQE energy error for TFIM instances with 7--14 qubits. Green bars show the error removable by optimal phase corrections, $E_{\mathrm{VQE}}-E_+(q)$. Blue bars show the residual error at fixed Born probabilities, $E_+(q)-E_{\mathrm{gs}}$. Black markers indicate their sum, the total VQE energy error.}
    \label{fig:phase_amplitude_analysis}
\end{figure}

Figure~\ref{fig:phase_amplitude_analysis} shows a substantial phase-removable contribution at every sampled system size. A nonzero residual remains after optimal phase correction, demonstrating the restriction imposed by the unchanged Born probabilities. Nevertheless, these instances retain considerable room for energy reduction through phase-only post-processing implemented by U-VQNHE. The phase-removable contribution quantifies the available improvement, where the fraction recovered in practice depends on the neural phase parametrization and its optimization.

\section{Stability of neural-network training}
\label{app:neuralnetwork}

All numerical experiments in the main text use sequential training: the VQE circuit is optimized first, and the neural network is then trained with the circuit parameters fixed. Here, we compare sequential and joint training under finite-shot sampling and examine the sensitivity of neural-network training to initialization and gradient behavior.

In VQNHE, joint updates change both the circuit's measurement distribution and the neural reweighting factors. Finite-shot fluctuations in the resulting normalization estimate can distort the energy objective and complicate optimization. U-VQNHE instead uses phase factors of unit modulus and has no sampled normalization denominator, removing this source of instability during joint training.

We compare two training schedules for the 14-qubit TFIM using 512 shots per measurement circuit. In sequential training, the VQE circuit is optimized first and its parameters are then fixed while the neural network is trained. In joint training, the circuit and neural parameters are updated together at each epoch, starting from an untrained circuit rather than the pretrained VQE state. Figure~\ref{fig:joint_training} shows the resulting finite-shot energy estimates over 500 neural or joint training epochs; the pretrained VQE energy and exact ground-state energy are shown for reference. Because the schedules start from different circuit states, this comparison describes their behavior as implemented rather than isolating the effect of the update schedule alone.

\begin{figure*}[t]
    \centering
    \begin{subfigure}[t]{0.49\textwidth}
        \centering
        \includegraphics[width=\linewidth]{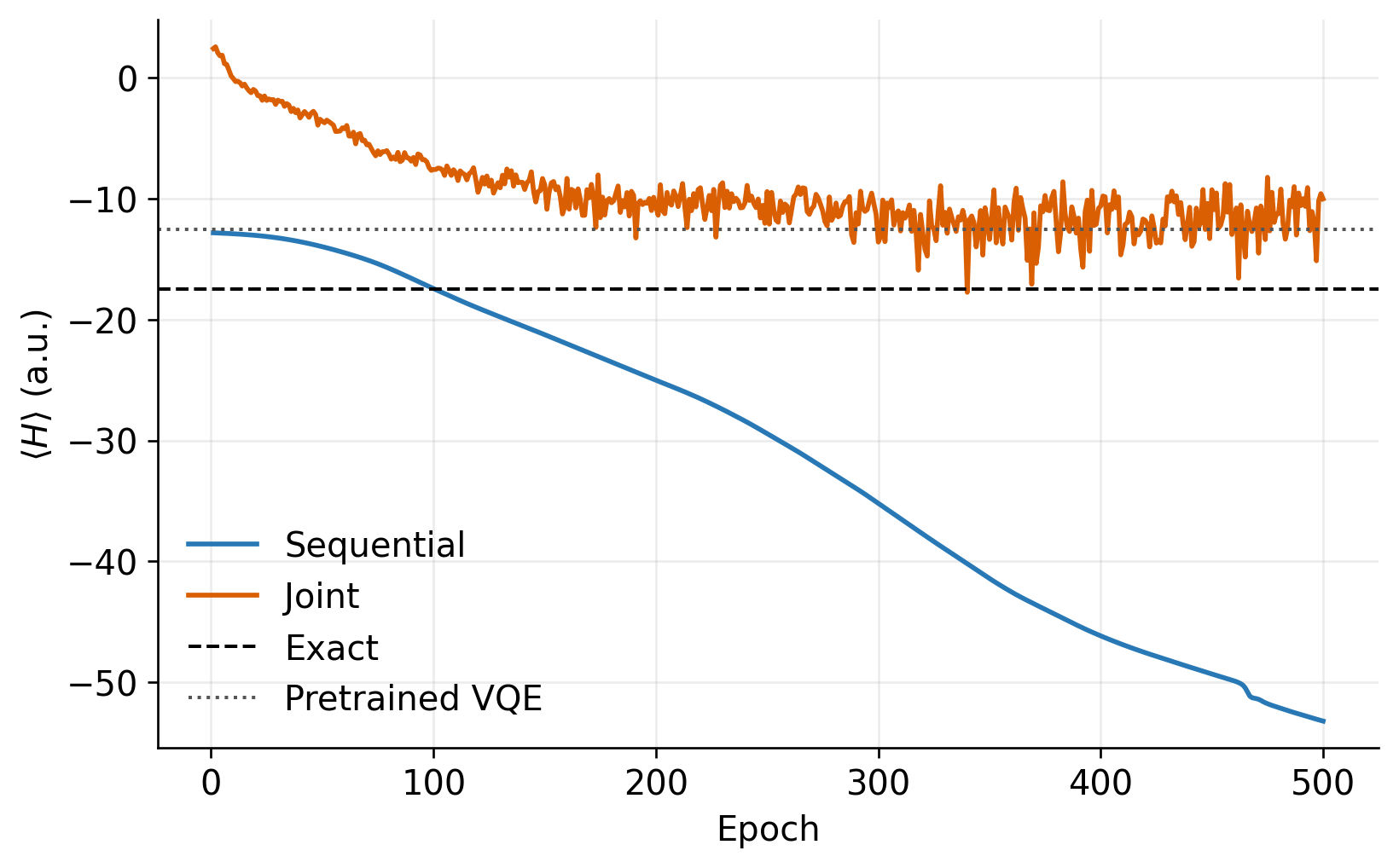}
        \caption{}
        \label{fig:joint_training_vqnhe}
    \end{subfigure}
    \hfill
    \begin{subfigure}[t]{0.49\textwidth}
        \centering
        \includegraphics[width=\linewidth]{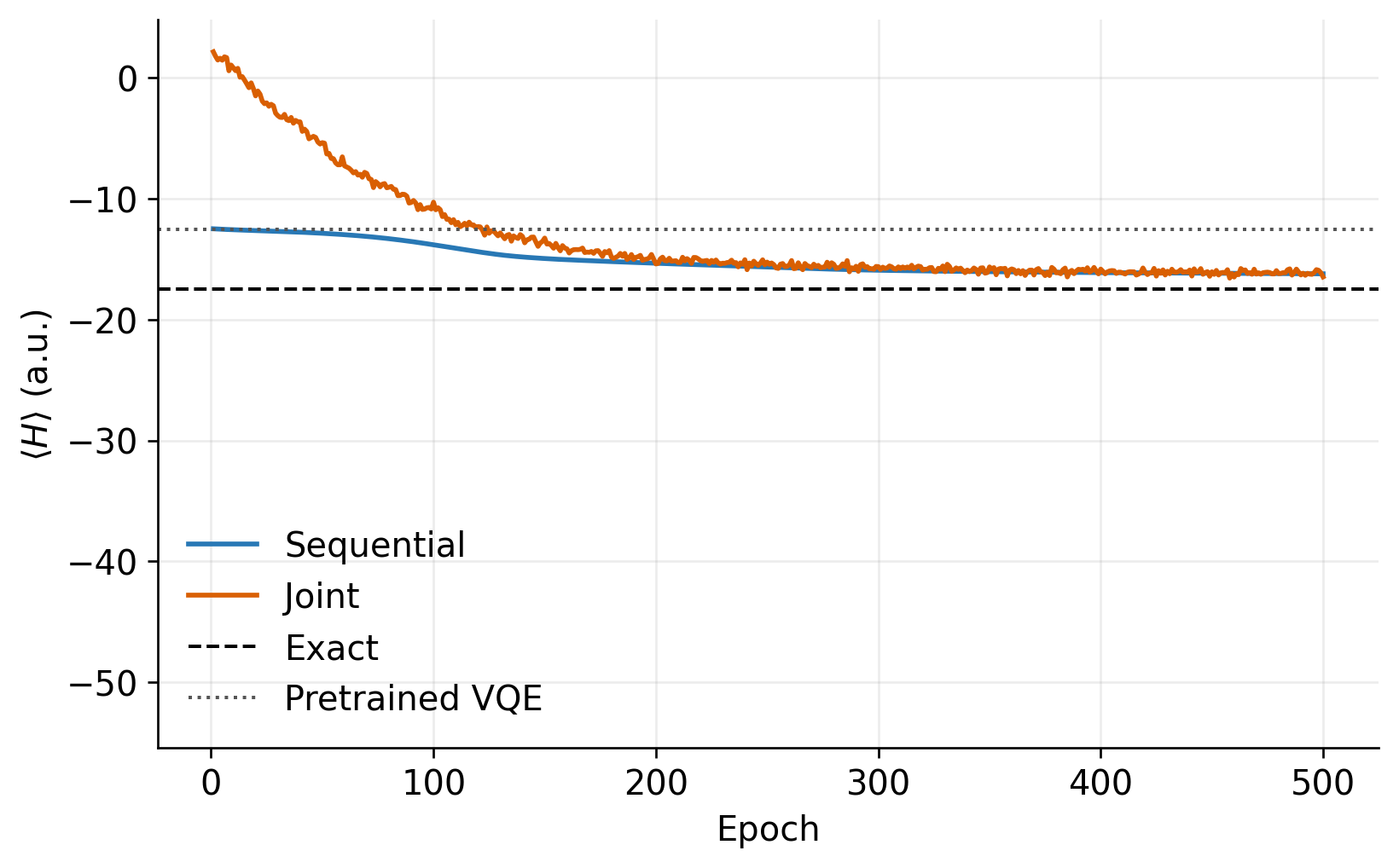}
        \caption{}
        \label{fig:joint_training_uvqnhe}
    \end{subfigure}
    \caption{\justifying Training trajectories for (a) VQNHE and (b) U-VQNHE on a 14-qubit TFIM, using 512 shots per measurement circuit. Sequential training fixes the pretrained VQE circuit, while joint training updates circuit and neural parameters together. Dotted and dashed lines indicate the pretrained VQE and exact ground-state energies, respectively.}
    \label{fig:joint_training}
\end{figure*}

The two schedules produce different behavior for VQNHE. Under sequential training, the estimate decreases to $-53.22$, far below the exact ground-state energy of $-17.47$. Under joint training, the estimate fluctuates above the pretrained VQE energy and ends at $-10.01$, showing no sustained improvement. The sub-variational sequential value is a finite-shot artifact, not a physical energy improvement. Together, the trajectories show that measurement noise can hinder joint optimization and distort the normalized energy estimate even when the circuit is fixed.

U-VQNHE shows a different trend. Sequential and joint training end at estimated energies of $-16.21$ and $-16.47$, respectively, a difference of $0.26$; both improve upon the pretrained VQE reference. Because the learned phase factors have unit modulus, U-VQNHE does not require a parameter-dependent normalization estimate. Finite-shot fluctuations therefore cannot be amplified by a small sampled denominator, consistent with the more stable trajectories in Fig.~\ref{fig:joint_training}. These results describe the tested instance and do not rule out other sources of optimization difficulty.

In these experiments, we optimize the circuit parameters with the gradient-free COBYLA method~\cite{Powell1994} and the neural-network parameters with Adam. Gradient-based circuit optimization is also possible: the parametrized Pauli rotations used in the ans\"atze admit the standard parameter-shift rule, which evaluates a circuit derivative from shifted circuit measurements~\cite{Wierichs2022}.

\begin{figure*}[t]
    \centering
    \begin{subfigure}[t]{0.49\textwidth}
        \centering
        \includegraphics[width=\linewidth]{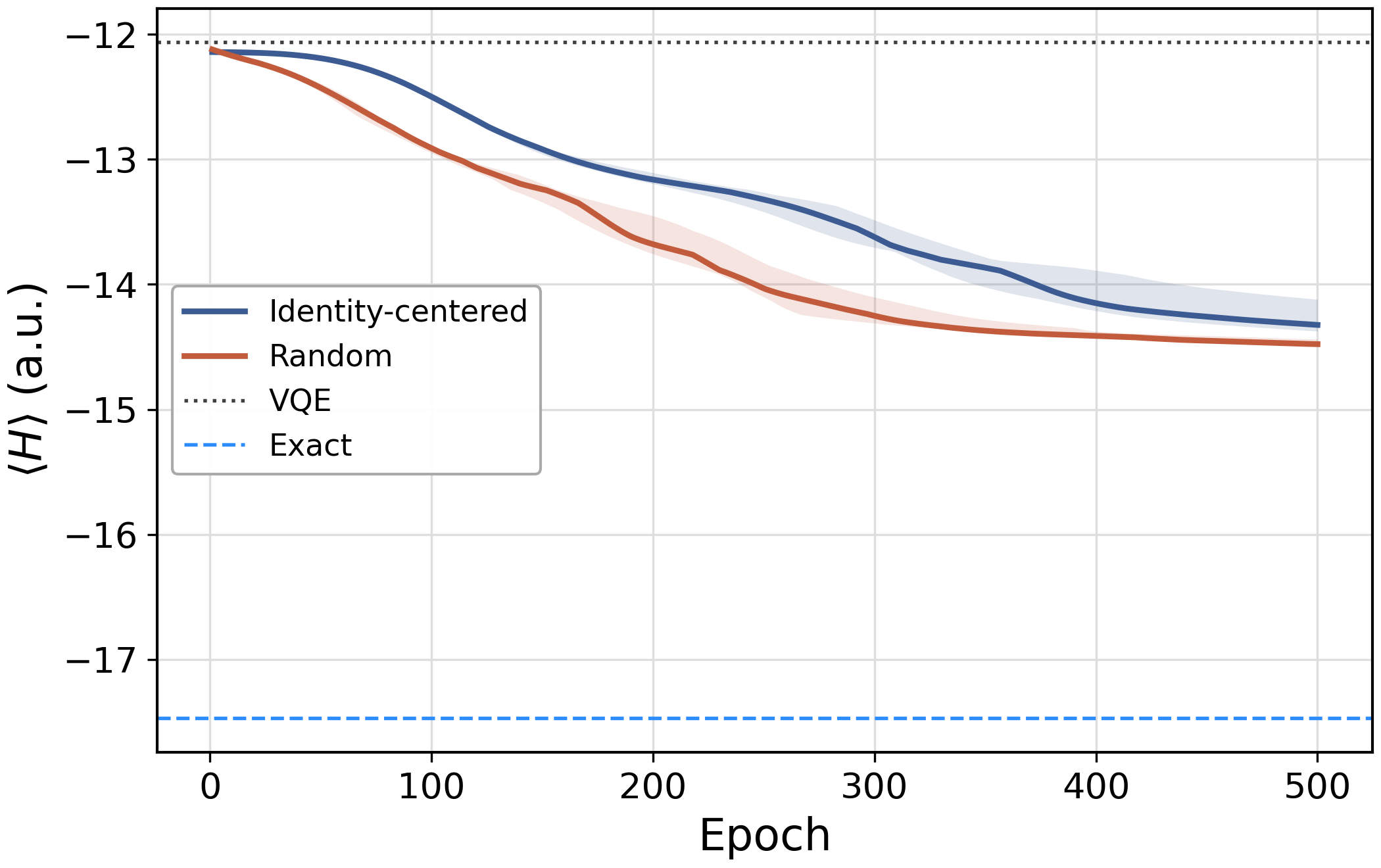}
        \caption{}
        \label{fig:initialization_energy}
    \end{subfigure}
    \hfill
    \begin{subfigure}[t]{0.49\textwidth}
        \centering
        \includegraphics[width=\linewidth]{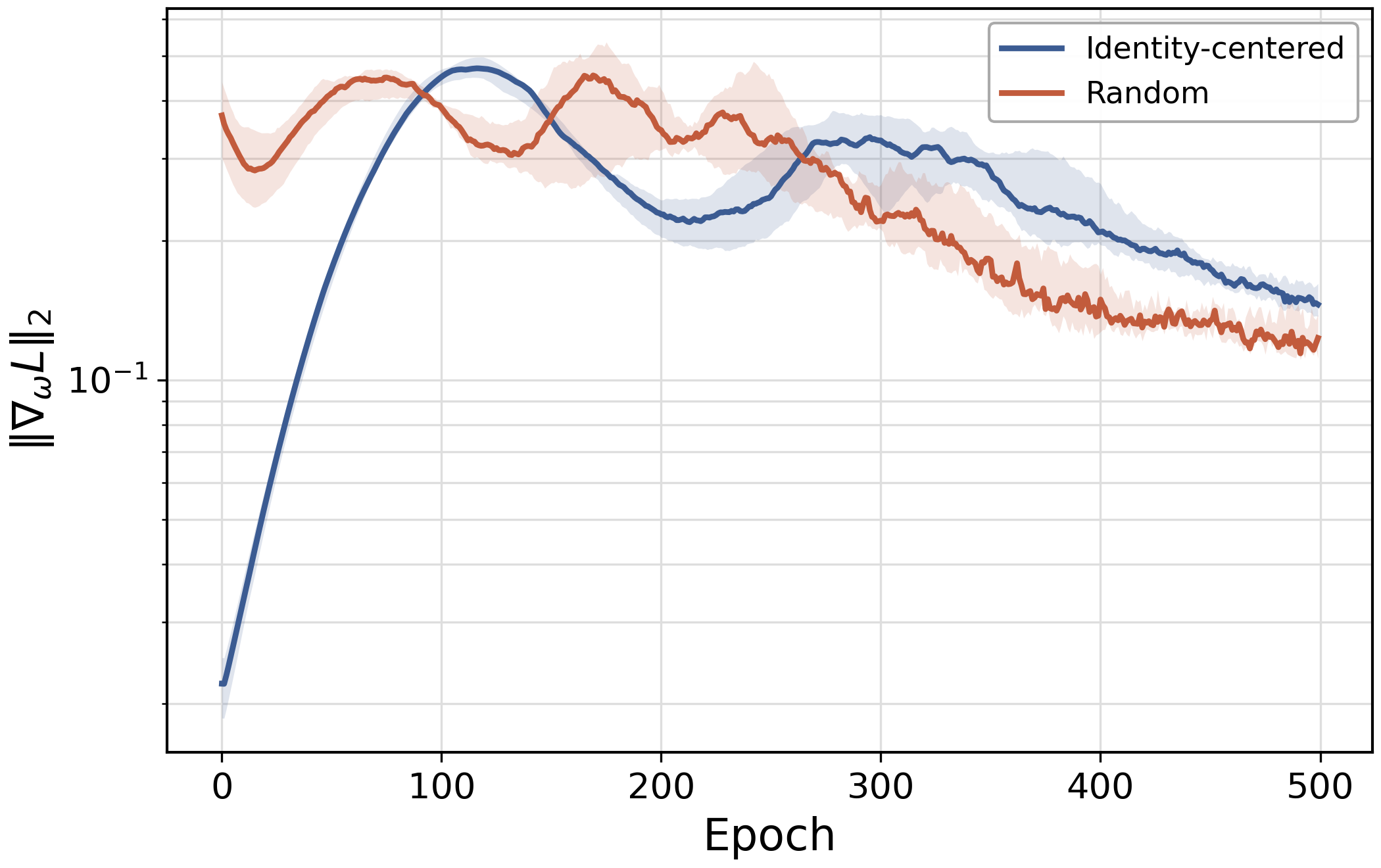}
        \caption{}
        \label{fig:gradient_norms}
    \end{subfigure}
    \caption{\justifying Sequential U-VQNHE training on the 14-qubit TFIM across ten independent trials. (a) Finite-shot energy trajectories. The dotted line marks the pretrained VQE energy, and the dashed line marks the exact ground-state energy. (b) Neural-network gradient norms. Blue curves correspond to networks initialized as the identity phase transformation, and red curves show independently randomized initialization. Shaded regions indicate variation between trials.}
    \label{fig:training_stability}
\end{figure*}

We further assess the stability of U-VQNHE across neural-network initializations using ten independent sequential-training trials. Figure~\ref{fig:training_stability}(a) shows that both identity-centered and random initializations lower the finite-shot energy below the frozen VQE reference and reach similar final values. Random initialization leads to a faster initial decrease and a slightly lower final estimate. The narrow spread between trials at later epochs indicates that initialization affects the training rate more strongly than the final result in this instance. The small variation among independently randomized trials further suggests that training is relatively insensitive to the particular neural-network initialization in this instance.

We also report the gradient norm $\|\nabla_{\boldsymbol\omega}\mathcal L_t\|_2$, where $\boldsymbol\omega$ denotes the neural-network parameters and $\mathcal L_t$ is the finite-shot energy objective at epoch $t$. The norms initially rise and then decrease for both initialization schemes [Fig.~\ref{fig:training_stability}(b)]. They show neither sustained growth nor a premature collapse to zero. Together, the energy and gradient trajectories provide evidence that the U-VQNHE neural network trains consistently across these initializations in the tested 14-qubit TFIM. This empirical stability complements the exact Rayleigh--Ritz guarantee of the represented state.

\bibliography{references}

\end{document}